\documentclass[11pt]{article}
\usepackage{amsmath,amssymb,amsthm} 
\usepackage{mathrsfs}
\usepackage{bbm}
\usepackage{stmaryrd}
\usepackage{paralist}
\usepackage{hyperref}
\usepackage{graphicx}
\usepackage{subcaption}
\graphicspath{{figures/}}
\usepackage{color}
\usepackage{todonotes}
\usepackage{bm}
\usepackage{float}

\setdefaultenum{(i)}{}{}{}
\usepackage[margin=1in]{geometry} 
\numberwithin{equation}{section}

\renewenvironment{thebibliography}[1]{%
\begin{oldthebibliography}{#1}%
\setlength{\baselineskip}{.8em}
\linespread{.8}
\small
\setlength{\parskip}{0ex}%
\setlength{\itemsep}{.1em}%
}%
{%
\end{oldthebibliography}%
}

\theoremstyle{plain}
\newtheorem{theorem}{Theorem}[section]
\newtheorem{proposition}[theorem]{Proposition}
\newtheorem{lemma}[theorem]{Lemma}
\newtheorem{corollary}[theorem]{Corollary}

\newtheorem{definition}[theorem]{Definition}

\DeclareMathAlphabet\scr{U}{scr}{m}{n}
\SetMathAlphabet\scr{bold}{U}{scr}{b}{n}
  \DeclareFontFamily{U}{scr}{\skewchar\font'177}%
  \DeclareFontShape{U}{scr}{m}{n}{<-6>rsfs5<6-8>rsfs7<8->rsfs10}{}%
  \DeclareFontShape{U}{scr}{b}{n}{<-6>rsfs5<6-8>rsfs7<8->rsfs10}{}%

\theoremstyle{definition}
\newtheorem{remark}[theorem]{Remark}
\newtheorem{assumption}[theorem]{Assumption}

\def\E{\mathbb{E}}
\def\R{\mathbb{R}}

\def\diag{\text{diag}}
\def\tr{\text{tr}}
\def\det{\text{det}}

\begin{document}

\title{An Equilibrium Model for the Cross-Section of Liquidity Premia\footnote{The authors are grateful to Jean-Philippe Bouchaud, Ibrahim Ekren, Martin Herdegen and Robert Pego for fruitful discussions, and to Steven E.~Shreve for pertinent remarks on an earlier version of the manuscript.}}

\author{Johannes Muhle-Karbe\thanks{Imperial College London, Department of Mathematics, email \url{j.muhle-karbe@imperial.ac.uk}.  Research supported by the CFM-Imperial Institute of Quantitative Finance.}
\and
Xiaofei Shi\thanks{Columbia University, Department of Statistics, email \url{xs2427@columbia.edu}.}
\and
Chen Yang\thanks{Chinese University of Hong Kong, Department of Systems Engineering and Engineering Management, email \url{cyang@se.cuhk.edu.hk}.}
}

\date{November 26, 2020}

\maketitle

\begin{abstract}
We study a risk-sharing economy where an arbitrary number of heterogenous agents trades an arbitrary number of risky assets subject to quadratic transaction costs. For linear state dynamics, the forward-backward stochastic differential equations characterizing equilibrium asset prices and trading strategies in this context reduce to a system of matrix-valued Riccati equations. We prove the existence of a unique global solution and provide explicit asymptotic expansions that allow us to approximate the corresponding equilibrium for small transaction costs. These tractable approximation formulas make it feasible to calibrate the model to time series of prices and trading volume, and to study the cross-section of liquidity premia earned by assets with higher and lower trading costs. This is illustrated by an empirical case study.\\

\emph{This paper is dedicated to the memory of our dear colleague Mark H.A. Davis, whose seminal works~\cite{davis.norman.90,davis.al.93} ushered in the mathematical analysis of models with transaction costs.}
\end{abstract}




\section{Introduction}

In the capital asset pricing model and many of its descendants, agents earn \emph{risk premia} for holding assets whose payoffs are uncertain. A number of influential empirical studies~\cite{amihud.mendelson.86a,brennan.subrahmanyam.96,pastor.stambaugh.03} suggest that -- in reality -- agents are also compensated for holding securities that are difficult to trade. To wit, if one sorts assets based on various measures of liquidity, then the returns earned by portfolios composed of the more liquid ones are systematically lower than for portfolios of less liquid assets.

The theoretical underpinnings of these \emph{liqiuidity premia} have been studied in an active literature going back to the seminal work of~\cite{constantinides.86}. This paper (and many more recent studies) takes a partial equilibrium approach, where the asset price dynamics are specified endogenously. Liquidity premia then refer to the amount by which the risky assets' expected returns have to be increased compared to a hypothetical frictionless version of the asset, in order  to offset the utility losses caused by the costs of trading. 

Another strand of research derives equilibrium asset prices with transaction costs endogenously by matching supply and demand~\cite{herdegen.al.20,isaenko.20,lo.al.04,sannikov.skrzypacz.16,vayanos.98,vayanos.vila.99,weston.18}. This allows to study how changes in liquidity feed back into asset prices, e.g., how liquidity premia are affected by the reduction of the fees charged by an exchange or the introduction of a financial transaction tax.

Yet, equilibrium models with a single illiquid risky asset still cannot say anything about the \emph{cross section} of liquidity premia across a spectrum of different assets -- that is, the subject of the empirical work of~\cite{amihud.mendelson.86a,brennan.subrahmanyam.96,pastor.stambaugh.03}. Equilibrium models with several illiquid assets lead to formidable computational challenges. These difficulties are of course only exacerbated if one moves beyond two (representative) agents that are typically assumed for tractability.\footnote{An alternative class of tractable models considers ``overlapping generations'' of agents that buy the securities when born and then either sell them after a prespecified holding time~\cite{acharya.pedersen.05} or gradually (and following a deterministic trajectory) over their lifetime~\cite{vayanos.98}.} To wit, even the most tractable models with linear state dynamics and quadratic transaction costs~\cite{garleanu.pedersen.16,isaenko.20,sannikov.skrzypacz.16} then lead to coupled systems of matrix Riccati equations. Whereas general well-posedness results are available for partial equilibrium models~\cite{annkirchner.kruse.15,bank.al.17,garleanu.pedersen.16,kohlmann.tang.02} or for models with exogenously given constant volatility~\cite{bouchard.al.18}, the only known results concerning the existence of equilibrium prices require the restrictive assumption that the agents' preferences are sufficiently similar~\cite{herdegen.al.20}, even in the case of only a single illiquid asset and just two agents. 

In the present study, we establish the existence of equilibrium prices for an arbitrary number of illiquid risky assets that are traded by an arbitrary number of agents. These agents have mean-variance preferences as in \cite{garleanu.pedersen.13,garleanu.pedersen.16} and trade to share the risk inherent in the fluctuations of their endowment streams, subject to a deadweight quadratic transaction cost as in \cite{almgren.chriss.01,garleanu.pedersen.13,garleanu.pedersen.16}. For assets that pay exogenous liquidating dividends at a finite terminal time, the ``Radner equilibrium'' where the agents act as price takers then can be characterized by a fully-coupled system of forward-backward stochastic differential equations (FBSDEs). If the terminal dividends and the volatilities of the agents' endowment streams are linear in the driving Brownian motions, then this FBSDE system can be reduced to a fully-coupled system of matrix-valued ordinary differential equations of Riccati form.

For the simplest case of a single risky asset traded by two agents, existence for this system has been established using Picard iteration by~\cite{herdegen.al.20}. However, even in this low-dimensional setting, establishing the convergence of the iteration scheme requires the restrictive assumption that the agents' risk aversions are sufficiently similar. In this paper, we show that this assumption is superfluous, in that the matrix Riccati system has a unique global solution even for an arbitrary number of agents and risky assets. 

In order to facilitate the calibration of the model to time-series data, we complement this main result with rigorous asymptotic expansions. In the practically relevant limiting regime of small transaction costs, this leads to explicit formulas for the impact of illiquidity on price levels, volatilities, and the cross section of liquidity premia that are earned by assets with different trading costs.

To bring these theoretical results to life, we test them using an empirical case study following~\cite{acharya.pedersen.05}. To wit, we sort the large-cap stocks in the S\&P index by Amihud's ``ILLIQ'' measure for liquidity~\cite{amihud2002illiquidity}, leading to three risky portfolios with high, medium, and low liquidity. In the frictionless version of our model, equilibrium returns solely compensate for risk and turn out to be very similar for all three portfolios. Using our asymptotic expansions, the calibration of the frictional version of the model to time series of prices \emph{and} trading volumes is still feasible. When trading costs are taken into account, the equilibrium returns of the high-liquidity portfolio are indeed decreased in line with the data, whereas their counterpart for the low liquidity portfolio are increased. However, to match the magnitude of the liquidity premia observed empirically in our model, the risk aversion coefficients of the agents need to be rather heterogenous. In line with the partial equilibrium literature, this suggests that additional features such as market closure \cite{dai2016portfolio}, unobservable regime shifts~\cite{chen2020incomplete}, or state-dependent trading costs~\cite{acharya.pedersen.05,lynch.tan.11} also play an important role in this context.

The remainder of this article is organized as follows. The exogenous inputs of the model are introduced in Section~\ref{s:model}. Subsequently, the frictionless version of the model is discussed in Section~\ref{sec:frictionless}. Section~\ref{sec:frictional} then contains our main results on the characterization of equilibrium prices and trading strategies with transaction costs. Their asymptotic expansions for small costs are collected in Section~\ref{sec:asymptotics}, and the model is calibrated to time-series data in Section~\ref{sec:calibration}. For better readability, all proofs are delegated to Section~\ref{sec:proofs}. 

\paragraph{Notation}

Throughout, we fix a filtered probability space $(\Omega,\mathscr{F},(\mathscr{F}_t)_{t \in [0,T]},\mathbb{P})$ with finite time horizon $T>0$, supporting a $D$-dimensional standard Brownian motion $(W_t)_{t \in [0,T]}$. 
For $p\geq 1$, we write $L^p(\R^m)$ for the $\R^m$-valued random variables $X$ satisfying $||X||_p:=\mathbb{E}[||X||^p]^{1/p}<\infty$ and denote by $\mathbb{H}^p(\mathbb{R}^{m\times n})$ the $\R^{m\times n}$-valued, progressively measurable processes $X=(X_t)_{t \in [0,T]}$ that satisfy 
\begin{align*}
\|X\|_{\mathbb{H}^p} :=\left(\E\left[\left(\int_0^T ||X_t||^2 dt\right)^{p/2}\right]\right)^{1/p}<\infty.
\end{align*} 
Here, for any vector or matrix, $||\cdot||$ is the Frobenius norm, i.e., the square root of the sum of squared entries. For $p \geq 1$, $\mathcal{S}^p(\R^m)$ denotes the $\R^m$-valued, progressively measurable processes $X=(X_t)_{t \in [0,T]}$ with continuous paths for which $\sup_{0 \leq t \leq T}||X_t|| \in L^p(\R^m)$. 

Finally, we write $\mathbbm{1}_m$ for the all-ones vector in $\R^m$ and $I_m$ for the identity matrix in $\R^{m\times m}$; the Kronecker product of matrices $A \in \mathbb{R}^{m \times n}$ and $B \in \mathbb{R}^{m' \times n'}$ is denoted by 
$$
A \otimes B :=\begin{bmatrix} A_{11}B & \cdots & A_{1n}B \\ \vdots  & \ddots & \vdots \\ A_{m1}B & \cdots & A_{mn}B \end{bmatrix} \in \mathbb{R}^{mm' \times nn'}, 
$$ 
and the Riemannian mean of two symmetric and positive definite matrices $A, B \in \mathbb{R}^{m\times m}$ is denoted by
$$
A \# B :=  A^{1/2} (A^{-1/2} B A^{-1/2})^{1/2} A^{1/2}.  
$$

\section{Agents, Endowments, and Financial Market}\label{s:model}
We consider $N \geq 2$ agents indexed by $n=1,2,\ldots, N$ who receive (cumulative) random endowments\footnote{An additional finite-variation drift would not affect the optimizers and hence the equilibrium prices  due to the mean-variance form of the optimization problems~\eqref{eq:nocosts} and \eqref{eq:costs} below.}
\begin{align}\label{eq:endowment}
d\zeta^n_t =({\xi^n_t})^\top dW_t, \quad \mbox{where } \xi^n \in \mathbb{H}^2(\mathbb{R}^{D }). 
\end{align}
To simplify the analysis below, we follow~\cite{lo.al.04} and assume that the agents' aggregate endowment is zero ($\sum_{n=1}^N\xi^n=0$). 

To hedge against the fluctuations of their endowment streams driven by the $D$-dimensional Brownian motion, the agents trade a safe and $K \leq D$ risky assets. 
The price of the safe asset is exogenous and normalized to one. The prices of the risky assets have dynamics
\begin{align}\label{eq:Sdyn}
dS_t=\mu_t dt+\sigma_t dW_t, \qquad S_T = \mathfrak{S}.
\end{align}
Here, the liquidating dividend $\mathfrak{S} \in L^2(\mathbb{R}^{K})$ is given exogenously. In contrast, the expected returns process $\mu \in \mathbb{H}^4(\mathbb{R}^K)$ and the volatility process $\sigma \in \mathbb{H}^4(\mathbb{R}^{K \times D})$ are to be determined endogenously by matching the agents' demand to the fixed supply $s\in\R^K$ of the risky assets.

\section{Frictionless Optimization and Equilibrium}\label{sec:frictionless}

As a benchmark, we first consider the frictionless version of the model. Starting from fixed initial positions $\varphi^n_{0-} \in \mathbb{R}^K$, $n=1,\ldots,N$ that clear the market ($\sum_{n=1}^N \varphi^n_{0-}=s$), the agents choose their positions $(\varphi_t)_{t \in [0,T]}$ in the risky assets to maximize one-period expected returns penalized for the corresponding variances as in \cite{bouchaud.al.12,garleanu.pedersen.13,garleanu.pedersen.16, kallsen.02,martin.14,martin.schoeneborn.11}. Without transaction costs, the continuous-time version of this criterion is
\begin{align}
\bar{J}_T^n(\varphi) &=\E\left[\int_0^T (\varphi_t^\top dS_t+d\zeta^n_t)-\frac{\gamma^n}{2}d\langle \textstyle{\int_0^\cdot} \varphi_u^\top dS_u+\zeta^n\rangle_t \right] \notag\\
&=\E\left[\int_0^T \Big(\varphi_t^\top \mu_t -\frac{\gamma^n}{2}\|\sigma_t^\top\varphi_t+\xi^n_t\|^2\Big)dt\right]. \label{eq:nocosts}
\end{align}
Here, $\gamma^n>0$ is the risk aversion of agent $n$; we assume without loss of generality that
\begin{align}\label{ordered risk aversion}
\gamma^N = \max \{\gamma^1, \ldots, \gamma^N\}. 
\end{align}
To ensure that the goal functional~\eqref{eq:nocosts} is well defined for any price dynamics~\eqref{eq:Sdyn} with $\mu \in \mathbb{H}^2(\mathbb{R}^K)$, $\sigma \in \mathbb{H}^4(\mathbb{R}^{K \times D})$, we focus on \emph{admissible} strategies $\varphi \in \mathbb{H}^4(\mathbb{R}^K)$.\footnote{The precise notion of admissibility is not crucial. We just need to ensure that the local martingale part of the wealth process $\int_0^\cdot \varphi_t dS_t$ is a true martingale.} Given that the covariance matrix $\sigma_t \sigma_t^\top \in \R^{K\times K}$ is invertible for every $t \in [0,T],$\footnote{This will be inherited from the terminal condition $\mathfrak{S}$ in the equilibrium we construct below.} each agent's optimal strategy for the frictionless problem~\eqref{eq:nocosts} is readily determined by pointwise optimization as 
\begin{align}\label{eq:strat}
\varphi^n_t = \big(\sigma_t\sigma_t^\top\big)^{-1} \frac{\mu_t}{\gamma^n} -\big({\sigma_t\sigma_t^\top}\big)^{-1}{\sigma_t\xi^n_t}, \quad t \in [0,T].
\end{align}

We are interested in ``competitive'' Radner equilibria~\cite{radner.72}, where each (small) agent takes the price dynamics of the risky assets as given in their individual optimization problem~\eqref{eq:nocosts}:

\begin{definition}\label{frictionless equilibrium}
A price process~\eqref{eq:Sdyn} for the risky assets is called a \emph{(Radner) equilibrium} if:
\begin{enumerate}
\item[i)] (\emph{Individual Optimality}) the corresponding individual optimization problem~\eqref{eq:nocosts} has a solution $\varphi^n$ for each agent $n=1,\ldots,N$;
\item[ii)] (\emph{Market Clearing}) the agents' total demand matches the supply of the risky assets at all times, in that $\sum_{n=1}^N \varphi^n_t=s$ for all $t \in [0,T]$.
\end{enumerate}
\end{definition}

For any equilibrium price $\bar{S}$ with dynamics~\eqref{eq:Sdyn}, matching the sum of the agents' corresponding demands \eqref{eq:strat} to the supply $s$ requires the following relation between the equilibrium expected returns $\bar{\mu}_t$ and volatility matrix $\bar{\sigma}_t$:
\begin{equation}\label{eq:eqnocosts}
\bar\mu_t= \bar{\gamma}\bar{\sigma}_t\bar{\sigma}_t^\top s, 
\quad t \in [0,T], \quad \mbox{where } \bar{\gamma}=\left(\sum_{n=1}^N \frac{1}{\gamma^n}\right)^{-1}.
\end{equation}
Together with the terminal condition from~\eqref{eq:Sdyn}, it follows that equilibrium prices correspond to solutions of the following system of quadratic backward stochastic differential equations (BSDEs):
\begin{equation}\label{eq:BSDES}
d\bar{S}_t=\big(\bar{\gamma}\bar{\sigma}_t\bar{\sigma}_t^\top s\big)dt+\bar{\sigma}_t dW_t, \qquad S_T = \mathfrak{S}.
\end{equation}
If the terminal condition $\mathfrak{S}$ 
is linear in the driving Brownian motion, then the BSDE~\eqref{eq:BSDES} can be solved explicitly, leading to an equilibrium price with Bachelier dynamics.

\begin{assumption}\label{ass:linear}
The terminal dividend is of the linear form
$$
\mathfrak{S}=\alpha W_T + \beta T, \quad \mbox{for $\beta\in\R^{D}$ and $\alpha\in \R^{K\times D}$ with $\mathrm{rank}(\alpha)=\mathrm{rank}(\alpha\alpha^\top)=K$.}
$$
\end{assumption}

\begin{proposition}\label{thm:frictionless}
Under Assumption~\ref{ass:linear}, a solution of the BSDE system~\eqref{eq:BSDES} and in turn a frictionless equilibrium price is given by 
\begin{align}\label{eq:Bachelier}
d\bar{S}_t = \big(\bar{\gamma}\alpha \alpha^\top s\big) dt + \alpha dW_t, \qquad \bar{S}_0 =  \big(\beta- \bar{\gamma}\alpha\alpha^\top s\big)T.
\end{align}
This equilibrium is unique among price dynamics with uniformly bounded volatility.
\end{proposition}

\section{Frictional Optimization and Equilibrium}\label{sec:frictional}

Now suppose as in \cite{almgren.chriss.01,garleanu.pedersen.13,garleanu.pedersen.16} that trading incurs quadratic costs on the turnover rate $\dot{\varphi}_t=d\varphi_t/dt$. The frictional analogue of the mean-variance goal functional~\eqref{eq:nocosts} then is
\begin{align}
J_T^n(\dot\varphi) =\E\left[\int_0^T \Big(\varphi_t^\top \mu_t -\frac{\gamma^n}{2}\|\sigma_t^\top\varphi_t+\xi_t^n\|^2 - \frac{1}{2}\dot{\varphi}_t^\top \Lambda \dot{\varphi}_t\Big)dt\right].\label{eq:costs}
\end{align}
Here, the transaction cost matrix $\Lambda$ is symmetric and positive definite,\footnote{As pointed out by~\cite{garleanu.pedersen.13}, symmetry of $\Lambda$ can be assumed without loss of generality because otherwise the symmetrized version $(\Lambda + \Lambda^\top)/2$ leads to the same trading costs. Positive definiteness means that each transaction has a positive cost. We write $\Lambda^{1/2}$ for the unique symmetric and positive definite square root of $\Lambda$, and note that $\Lambda$ and $\Lambda^{1/2}$ both are invertible. } and we focus on \emph{admissible} trading strategies that are absolutely continuous with rate $\dot{\varphi} \in \mathbb{H}^4(\mathbb{R}^K)$.\footnote{The corresponding positions then also automatically belong to $\mathbb{H}^4(\R^K)$ as in the frictionless case, so that the frictional goal functional is well defined for expected returns process $\mu \in \mathbb{H}^2(\R^K)$ and volatility matrix $\sigma \in \mathbb{H}^4(\R^{K \times D})$.}

\begin{remark}
As in \cite[Section~3.2]{garleanu.pedersen.16}, the deadweight transaction costs can be seen as a compensation paid to liquidity providers who intermediate between the agents we model in the present paper.  Non-trivial off-diagonal elements of $\Lambda$ then correspond to cross price impact due to each assets' contribution to the intermediaries' portfolio. Alternatively, if the quadratic costs are interpreted as more tractable proxies for linear costs such as bid-ask spreads or a transaction tax, then a diagonal matrix is the natural specification for $\Lambda$.
\end{remark}

With transaction costs, the agents' optimal strategies are no longer myopic. Instead, they are characterized by the first-order condition that the G\^ateaux derivative of the respective goal functionals~\eqref{eq:costs} vanishes for all perturbations of the trading rate. Together with Fubini's theorem, this yields
\begin{align}\label{eq:condition}
\hspace{-5pt}\Lambda \dot{\varphi}^n_t &= \E_t\Big[\int_t^T \big(\mu_u -\gamma^n\sigma_u(\sigma_u^\top\varphi^n_u+\xi^n_u) \big)du\Big]\notag
\\&= \E_t\Big[\int_0^T \big(\mu_u -\gamma^n\sigma_s(\sigma_s^\top\varphi^n_u+\xi^n_u) \big)du\Big] + \int_0^t \big(\gamma^n\sigma_u(\sigma_u^\top\varphi^n_u+\xi^n_u) -\mu_u\big)du.
\end{align}
To clear the market, the sum of all agents' trading rates has to vanish at all times. Therefore, after summing the agents' first-order conditions~\eqref{eq:condition}, both the martingale and the drift terms need to vanish for all $t \in [0,T]$. The frictional equilibrium return in turn has to satisfy 
\begin{align*}
0 
&= \sum_{n=1}^N \left(\mu_t -\gamma^n\sigma_t(\sigma_t^\top\varphi^n_t+\xi^n_t) \right).
\end{align*}
Taking into account the market clearing condition $\sum_{n=1}^N \varphi^n_t=s$ and recalling that the aggregate endowment is zero ($\sum_{n=1}^N \xi^n_t = 0$), the price dynamics~\eqref{eq:Sdyn} therefore again lead to a BSDE system for the equilibrium asset price:
\begin{align}\label{eq:Sdef}
dS_t= \left(\frac{\gamma^N}{N}\sigma_t \sigma_t^\top s + \frac{{\sigma}_t{\sigma}_t^\top}{N} \sum_{n=1}^{N-1} (\gamma^n-\gamma^N)\left({\sigma}_t^\top {\varphi}^n_t +\xi^n_t\right)\right)dt+\sigma_t dW_t, \quad S_T=\mathfrak{S}.
\end{align}
However, these equations are now no longer autonomous but coupled to the forward equations for the optimal positions, 
\begin{align}\label{eq:varphidef}
d\varphi^n_t=\dot{\varphi}^n_t dt, \qquad \varphi^n_{0} = \varphi^n_{0-}, \quad n=1,\ldots,N-1,
\end{align}
as well as the backward equations for the corresponding optimal trading rates $\dot{\varphi}^n_t$ implied by the first-order conditions~\eqref{eq:condition}:
\begin{align}\label{eq:BSDEZbis}
d\dot{\varphi}^n_t &= \Lambda^{-1}\big(\gamma^n\sigma_t(\sigma_t^\top \varphi^n_t+\xi^n_t) -\mu_t\big)dt + \dot{Z}^n_t dW_t \qquad\qquad  \dot{\varphi}^n_T = 0, \quad n=1,\ldots,N-1,\\
&=\textstyle{\Lambda^{-1}\sigma_t\big(\sigma^\top_t(\gamma^n \varphi^n_t - \frac{1}{N}\sum_{m=1}^{N-1} (\gamma^m - \gamma^N)\varphi^m_t)
+( \gamma^n\xi_t^n - \frac{1}{N} \sum_{m=1}^{N-1} (\gamma^m-\gamma^N)\xi_t^m) - \frac{\gamma^N }{N}\sigma^\top_t s\big)dt}\notag\\&\textstyle{\quad + \dot{Z}^n_t dW_t.} \notag
\end{align}
(The position and trading rate of agent $N$ are in turn pinned down by market clearing.) To express this forward-backward system more compactly in matrix-vector notation, we write 
\begin{align}
\varphi_t : = 
\begin{bmatrix}
\varphi^1_t \\
\vdots \\
\varphi^{N-1}_t
\end{bmatrix}, 
\qquad 
\dot{\varphi}_t : = 
\begin{bmatrix}
\dot{\varphi}^1_t \\
\vdots \\
\dot{\varphi}^{N-1}_t
\end{bmatrix}, 
\qquad
\dot{Z}_t : = \begin{bmatrix}
\dot{Z}^1_t\\
 \vdots\\
\dot{Z}^{N-1}_t
\end{bmatrix},
\qquad\xi_t : = \begin{bmatrix}
\xi^1_t\\
\vdots\\
\xi^{N-1}_t
\end{bmatrix},
\end{align}
and define the risk-aversion matrix
\begin{align}\label{def: gamma}
\Gamma := \diag\{\gamma^1, \cdots, \gamma^{N-1}\} - \frac{1}{N} \mathbbm{1}_{N-1}\mathbbm{1}^\top_{N-1} \diag\{\gamma^1 - \gamma^N, \cdots, \gamma^{N-1} - \gamma^N\} \in \mathbb{R}^{(N-1)\times (N-1)}.
\end{align}
The above discussion then can be summarized as follows:

\begin{lemma}\label{lem:FBSDE}
Suppose there exists a solution $(\varphi,\dot{\varphi},\dot{Z},S, \sigma) \in \mathbb{H}^4(\R^{K(N-1)}) \times \mathbb{H}^4(\R^{K(N-1)}) \times \mathbb{H}^2(\R^{K(N-1) \times D}) \times \mathcal{S}^2(\R^K) \times \mathbb{H}^4(\R^{K \times D})$ of the following FBSDE system:
\begin{alignat*}{2}
&d\varphi_t = \dot{\varphi}_t dt, \quad  && \varphi_0=\varphi_{0-},\\
&d\dot{\varphi}_t = \left((\Gamma\otimes\Lambda^{-1}\sigma_t {\sigma}_t^\top)\varphi_t + (\Gamma\otimes\Lambda^{-1}\sigma_t)\xi_t - \frac{\gamma^N}{N} \mathbbm{1}_{N-1}\otimes\Lambda^{-1}\sigma_t \sigma_t^\top s\right) dt + \dot{Z}_t dW_t, \quad  && \dot{\varphi}_T =0,\\
&dS_t =  \left(\frac{\gamma^N}{N}{\sigma}_t {\sigma}_t^\top s + \frac{{\sigma}_t}{N} \sum_{n=1}^{N-1} (\gamma^n-\gamma^N)\left({\sigma}_t^\top{\varphi}^n_t +\xi_t^n \right)\right)dt+\sigma_t dW_t, \quad  && S_T=\mathfrak{S}.
\end{alignat*}
Then, $S$ is a Radner equilibrium with transaction costs, in that the trading rates $\dot\varphi^1,\ldots,\dot\varphi^{N-1}$ and $\dot{\varphi}^N=-\sum_{n=1}^{N-1} \dot{\varphi}^n$ are optimal for the frictional optimization problems~\eqref{eq:costs} of agents $n=1,\ldots,N$, and clear the market.
\end{lemma}

For the simplest case of a single risky asset and two agents, the FBSDE system from Lemma~\ref{lem:FBSDE} has been studied by~\cite{herdegen.al.20}.\footnote{If one penalizes squared inventories rather than the corresponding fluctuations (which depend on the \emph{endogenous} volatility), then the FBSDE system becomes \emph{linear} and can be analyzed in very general settings, in particular, for arbitrary numbers of agents, compare~\cite{bank.al.18,muhlekarbe.al.20}.} More specifically, \emph{local} existence is established there under the restrictive condition that the agents' risk aversion coefficients are sufficiently similar.\footnote{If all agents have the same risk aversion coefficient, then the BSDE for the frictional equilibrium price decouples from the other components of the FBSDE system in Lemma~\ref{lem:FBSDE} and reduces to its frictionless counterpart similarly as in \cite{herdegen.al.20}. Similar risk aversions in turn lead to frictional equilibrium prices in the vicinity of their frictionless counterparts, so that existence can be established using a Picard iteration under smallness conditions inspired by~\cite{tevzadze.08}.} If the terminal condition $\mathfrak{S}$ is linear in the driving Brownian motion as in Assumption~\ref{ass:linear} and the volatilities $\xi^n_t$ of the agents' endowments are of the same linear form, then the FBSDE system can be reduced to a system of Riccati equations by an appropriate ansatz. However, the system consists of four fully coupled equations even for a single risky asset and two agents, so that existence (established via Picard iteration) is again only known if the agents' preferences are sufficiently homogenous~\cite[Theorem~5.2]{herdegen.al.20}.\footnote{Equilibria in linear-quadratic models are also linked to systems of nonlinear equations in~\cite{isaenko.20,sannikov.skrzypacz.16}, but the existence of a unique solution is left open in these studies.} These difficulties are of course only exacerbated for multiple assets and agents, because each of the Riccati equations becomes matrix valued in this case.

In the present paper, we overcome these difficulties and establish \emph{global} existence for the FBSDE system from Lemma~\ref{lem:FBSDE} for linear terminal conditions and endowment volatilities:

\begin{assumption}\label{ass:linear2}
The volatilities of the agents' endowment streams~\eqref{eq:endowment} are of the form  
$$
\xi^n_t=\xi^n W_t, \quad \mbox{for $\xi^n \in \mathbb{R}^{D\times D}$}.
$$ 
With a slight abuse of notation, we set $\xi=[\xi^1,\ldots,\xi^{N-1}]^\top \in \mathbb{R}^{(N-1)D\times D}$.
\end{assumption}

Like~\cite[Theorem 5.2]{herdegen.al.20}, our existence result in Theorem~\ref{thm:Radnerfric} exploits the link between the FBSDE system and a system of Riccati ODEs. However, to make the latter more amenable to analytical estimates, we perform a number of changes of variables that allow to reduce the number of coupled (matrix) equations from four to two. Standard comparison arguments still do not apply to this multidimensional system, in particular, when the equations are matrix-valued for many risky assets and agents. However, another reparametrization finally leads to a system where the right-hand side of one equation is linear in this component. A matrix version of the variation of constants formula in turn allows to derive bounds on the unique local solution of this equation. This in turn finally allow us to obtain global existence by applying Gronwall's inequality to a \emph{scalar} function -- the norm of the local solution of the other equation on the cone of positive semidefinite matrices.

To formulate these results, we first state our global wellposedness result for the reduced ODE system. (The proof is deferred to Section~\ref{proof of sec: frictional} for better readability.)

\begin{lemma}\label{thm:ivp}
Define
$$
c := \begin{bmatrix} c_1 & \cdots & c_{N-1} \end{bmatrix}^\top, \quad \mbox{where } c_n := \bar{\gamma}\left(\frac{1}{\gamma^n} - \frac{1}{\gamma^N}\right)>0.
$$
There exists a unique global solution $(F,H)$ on $[0,T]$ of the following initial value problem:
\begin{equation}
\label{eqs: Ricatti}
\left\{
\begin{aligned}
F'&=\Gamma\otimes \left(\alpha + \left({c} \otimes I_K\right)^\top H \right)\left(\alpha + \left({c} \otimes I_K\right)^\top H \right)^\top - F\left(I_{N-1}\otimes \Lambda^{-1}\right)F,  &F(0) = 0, \\
H' &= \left(\Gamma\otimes \left(\alpha + \left({c} \otimes I_K\right)^\top H\right)\right)\xi - F\left(I_{N-1}\otimes \Lambda^{-1}\right)H,  & H(0) =0. 
\end{aligned}
\right.
\end{equation}
Moreover, $F$ takes values in the positive semidefinite matrices.\footnote{In much of the literature, positive definite matrices are additionally required to be symmetric. This does not generally hold for $F$, however, so that the arguments below need to be developed without this convenient property.} 
\end{lemma}

With the solution of the matrix Riccati equations~\eqref{eqs: Ricatti} at hand, we can then construct a solution of the FBSDE system from Lemma~\ref{lem:FBSDE}. The latter in turn leads to a Radner equilibrium with transaction costs. (The proof is again delegated to Section~\ref{proof of sec: frictional} for better readability.)

\begin{theorem}\label{thm:Radnerfric}
With the functions $F$, $H$ from Lemma~\ref{thm:ivp}, let $\Phi(\tau)$ be the solution of the linear matrix ODE\footnote{This is the exponential of $\int_0^\cdot (I_{N-1}\otimes \Lambda^{-1/2}) F^\top(r) (I_{N-1}\otimes \Lambda^{-1/2}) dr$ in the scalar case or if the matrices involved commute.}
\begin{equation}\label{eq:matrix exponential}
\Phi'(t) = \left(I_{N-1}\otimes \Lambda^{-1/2}\right) F^\top(T-t) \left(I_{N-1}\otimes \Lambda^{-1/2}\right) \Phi(t), \quad \Phi (0) = I_{K(N-1)},
\end{equation}
and define 
\begin{align}\label{eq:state transition}
\Psi (r;t) :=  \left(I_{N-1}\otimes \Lambda^{1/2}\right)\Phi(r)\Phi^{-1}(t) \left(I_{N-1}\otimes \Lambda^{-1/2}\right), \quad  \mbox{for $r, t \in[0,T].$}
\end{align}
Suppose Assumptions~\ref{ass:linear} and \ref{ass:linear2} are satisfied. With the frictionless equilibrium price and volatility $(\bar{S}, \bar{\sigma})$ from Proposition~\ref{thm:frictionless}, a solution $({\varphi}, \dot{\varphi}, \dot{Z}, \bar{S}+\mathcal{Y}-(c\otimes \Lambda)^\top \dot{\varphi}, \bar{\sigma}-(c\otimes \Lambda)^\top \dot{Z})$ of the FBSDE system from Lemma~\ref{lem:FBSDE} is then given by
\begin{align}
\varphi_t &= \bar{\varphi}_0 +\Psi^\top(0;t)\left(\varphi_{0-} - \bar{\varphi}_0\right) 
- \int_0^t\Psi^\top(r;t) \left(I_{N-1}\otimes \Lambda^{-1}\right)H(T-r) W_r dr ,\label{eq:frictional position}\\
\dot{\varphi}_t &= - \left(I_{N-1}\otimes \Lambda^{-1}\right)\left[F(T-t) \left(\varphi_t - \bar{\varphi}_0\right) + H(T-t)W_t\right] , \label{eq:frictional strategy}\\
\mathcal{Y}_t &= -\bar{\gamma} \left( \int_0^{T-t}  \left(\left({c} \otimes I_K\right)^\top H \alpha^\top + \alpha H^\top\left({c} \otimes I_K\right) +\left({c} \otimes I_K\right)^\top H H^\top  \left({c} \otimes I_K\right)\right)(r) dr  \right) s\label{eq:Ybis},\\
\dot{Z}_t &= -\left(I_{N-1}\otimes \Lambda^{-1}\right)H(T-t). \label{eq:volatility}
\end{align}
In particular, $S=\bar{S}+\mathcal{Y}_t- (c \otimes \Lambda)^\top  \dot{\varphi}_t$  is a Radner equilibrium with transaction costs.
\end{theorem}

\section{Small-costs Asymptotics}\label{sec:asymptotics}

The Riccati system~\eqref{eqs: Ricatti} can be solved numerically using standard ODE solvers by vectorizing the matrix equations. In order to glean qualitative insights into the structure of the solution and facilitate the calibration of the model parameter to time series data, it is nevertheless instructive to expand the solution in the practically relevant limiting regime of \emph{small} transaction costs.  (Again, the proof of Theorem~\ref{thm:asymp} is deferred to Section~\ref{proof of asymptotics} for better readability.)

\begin{theorem}\label{thm:asymp}
Fix a positive definite matrix $\bar\Lambda$ and set\footnote{Note that the square root of the risk-aversion matrix $\Gamma$ is well defined by~\cite[Theorem 1.29]{higham2008functions}, even though this matrix is generally only positive semidefinite but not symmetric.}
$$
M:=\left({c} ^\top \Gamma^{1/2}\otimes \bar\Lambda\left(\bar\Lambda \# \alpha\alpha^\top \right)^{-1}\alpha\right).
$$
For small transaction costs $\Lambda=\lambda \bar\Lambda$ with $\lambda \to 0$, 
the difference between the frictional equilibrium volatility from Theorem~\ref{thm:Radnerfric} and its frictionless counterpart $\bar{\sigma}=\alpha$ from Proposition~\ref{thm:frictionless} has the following leading-order expansion: 
\begin{align}\label{eq: small cost delta sigma}
\int_0^T \|\sigma_t-\bar{\sigma}  -  \lambda^{1/2}M \xi \|_{\mathrm{op}} \; dt = O(\lambda).
\end{align}
For $\varphi_{0-} = \bar\varphi_0$,\footnote{As in \cite{moreau.al.17}, the same expansion remains valid if the initial condition is close enough to the frictionless allocation, which is a natural assumption for a market with small trading costs.} the leading-order adjustment of the initial price level can be approximated as 
\begin{align}\label{eq: small cost stock correction}
S_0-\bar{S}_0= - \lambda^{1/2}\bar{\gamma}\left( M \xi\alpha^\top  + \alpha\xi^\top M^\top  \right)   s T+O(\lambda).
\end{align}
Finally, the equilibrium expected returns satisfy
$$
\left\| \mu -\left(\bar{\mu}+ \Delta \bar{\mu} +\lambda^{1/2}\left( c^\top \Gamma^{1/2} \otimes \left(\bar\Lambda \# \alpha\alpha^\top\right) \right)\dot{\bar{\varphi}}\right)\right\|_{\mathbb{H}^p}=O(\lambda).
$$
Here, the average adjustment compared to the frictionless case are given by 
$$
\Delta\bar{\mu}:= \lambda^{1/2} \bar{\gamma} \left(M \xi\alpha^\top  + \alpha\xi^\top M^\top  \right) s=O(\lambda^{1/2}).
$$
The process $\dot{\bar{\varphi}}$, that describes the mean-zero fluctuations around this constant value, follows an ${K(N-1)}$-dimensional Ornstein-Uhlenbeck process:
\begin{align*}
d\dot{\bar{\varphi}}_t = -\lambda^{-1/2}\left( \Gamma^{1/2} \otimes \bar\Lambda^{-1}\left(\bar\Lambda\#\alpha\alpha^\top \right) \right) \left(\dot{\bar{\varphi}}_t dt + \left(I_{N-1}\otimes \left(\alpha\alpha^\top\right)^{-1}\alpha\right) \xi dW_t \right).
\end{align*}
This process also provides a leading-order approximation of the equilibrium (signed) trading volume, in that $\|\dot{\varphi}-\dot{\bar{\varphi}}\|_{\mathbb{H}^p}=O(1)$ for every $p>1$. 
\end{theorem}

These formulas simplify considerably in the case of two agents ($N=2$). To wit, the risk-aversion matrix $\Gamma$ and the risk-aversion vector $c$ then collapse to the scalars
$$
\Gamma=\frac{\gamma^1+\gamma^2}{2}, \quad c= \bar{\gamma}\frac{\gamma^2-\gamma^1}{\gamma^1\gamma^2}=\frac{\gamma^2-\gamma^1}{\gamma^1+\gamma^2}.
$$
As a result, the average adjustments of the expected returns compared to the frictionless case simplify to
\begin{align}\label{eq: asymp lp}
\lambda^{1/2} \bar\gamma \left( M \xi\alpha^\top  + \alpha\xi^\top M^\top  \right)   s,  ~~\text{ where } M=\frac{\gamma^2-\gamma^1}{\sqrt{2(\gamma^1+\gamma^2)}}  \bar\Lambda\left(\bar\Lambda \# \alpha\alpha^\top \right)^{-1}\alpha.
\end{align}
The corresponding leading-order approximation of the (signed) trading volume is
\begin{align}\label{eq: asymp phidot}
d \dot{\bar{\varphi}}_t &= -\lambda^{-1/2}\sqrt{\frac{\gamma^1+\gamma^2}{2}} \bar\Lambda^{-1} \left(\bar\Lambda\#\alpha\alpha^\top \right)\left(\dot{\bar{\varphi}}_t dt + \left(\alpha \alpha^\top\right)^{-1} \alpha \xi dW_t \right). 
\end{align}

These explicit formulas clearly separate the impact of risk, (heterogeneity of) risk aversions, trading costs, and individual trading motives. This is makes it feasible to calibrate the model to time series of prices and trading volume, as we discuss now.

\section{Calibration to Time-Series Data}\label{sec:calibration}

\subsection{Dataset}

Following empirical research of~\cite{amihud2002illiquidity,acharya.pedersen.05} and industry practice as documented in \cite{vanguard2018}, we study liquidity premia for US equities by constructing portfolios corresponding to different levels of liquidity. To wit, we build portfolios H, M and L, which correspond to High, Medium and Low liquidity, respectively, from 1991 to 2016. The portfolios are constructed in a tradable manner: for each portfolio, the number of shares in each constituent stock in year $T$ is computed using only the data in year $T-1$ (so the data in 1990 is used to calculate the portfolio weight in 1991, for example, in order to avoid forward-looking biases), and kept constant throughout year $T$. We choose this 26 year investment period to match the estimation in \cite{vanguard2018} on Russell indices.\footnote{Also note that an even longer period would be increasingly at odds with our arithmetic model and the Bachelier-type price dynamics it implies.} We obtain the S\&P500 constituents from 1990 to 2016 from Compustat, match them to the CRSP daily stock file based on the \texttt{CUSIP} identifier,\footnote{See \url{http://www.crsp.org/products/documentation/security-data}.} and then obtain the daily adjusted closing prices, trading volumes, and shares outstanding. 

The constituent stocks for each portfolio in year $T$ are selected as follows. First, we carry out a prescreening using the data in year $T-1$ similar to \cite{amihud2002illiquidity} to focus on stocks that (1) remain constituents for the whole year $T-1$, (2) have more than 200 trading days with available price data and positive volume in year $T-1$, and (3) have available prices on the first trading day of year $T$. Second, among these prescreened stocks, we pick the 200 stocks with the highest average daily market capitalization, the same number of stocks as the in large-cap portfolio considered in \cite{vanguard2018}. These 200 stocks are then sorted by their transaction costs proxied by ILLIQ in year $T-1$, and separated into three groups with 67, 66, and 67 stocks, respectively. Here, ILLIQ is a liquidity index proposed by \cite{amihud2002illiquidity}, defined as the average of the absolute value of daily percentage return of a stock divided by its dollar volume. A  higher ILLIQ value (i.e., large price moves even with little trading) of a stock indicates a lower liquidity level. Third, motivated by \cite{acharya.pedersen.05},\footnote{\cite{acharya.pedersen.05} observed liquidity premia for the equal-weighted returns of various portfolios. However, to achieve such returns in practice, these portfolios need to be rebalanced daily. We rebalance the portfolio at the beginning of each year to stay close to a buy-and-hold strategy, which seems natural given that the portfolios are interpreted as assets that can be bought and hold in our model.} for each group, we form a portfolio that is equal-weighted in year $T-1$ in the sense that all constituent stocks have equal values under the their respective average price in year $T-1$. In summary, this leads to three portfolios H, M and L with the lowest, medium, and highest transaction costs, respectively. 

We view these three portfolios as three risky assets with different liquidity. The trading volumes and outstanding shares for each portfolio are calculated as the aggregated values for all constituent stocks. On the first trading day, we set the price of each portfolio to be the average prices of constituents weighted by their shares outstanding, so that the price multiplied by the shares outstanding equals the total market capitalization for the constituents of each portfolio. The portfolio is then rebalanced at the beginning of each subsequent year. To determine the transaction cost associated with each portfolio, we first calculate the daily values as the equal-weighted average of the transaction cost of all constituent stocks on each day, and then calculate the average of these daily values during the whole sample period.

For our 26 years of data the average historical shares outstanding are $s=(1.15, 0.32, 0.23)^\top\times 10^{10}$; the average prices (in dollars) are $(45.41,49.23,38.30)^\top$. The annualized arithmetic return is $\hat{\mu}=(2.99,3.71,3.55)^\top$; dividing by the average prices, this corresponds to a (relative) Black-Scholes return of $(6.57\%,7.55\%,9.27\%)^\top$. In particular, the liquidity premium of the low-liquidity portfolio L compared to the high-liquidity portfolio H (i.e., the difference between the respective Black-Scholes returns) is 2.69\%, in line with the 2.4\% reported for Russell data in \cite{vanguard2018}. The corresponding estimate for the annualized arithmetic variance is 

\begin{align*}
\hat\Sigma=\begin{pmatrix}
   72.00  & 71.49  & 54.80 \\
   71.49  & 85.42  & 65.86 \\
   54.80  & 65.86  & 56.84 \\
\end{pmatrix}.
\end{align*}

\subsection{Calibration of the Frictionless Model}

We first consider the frictionless version of the model and check whether the liquidity premium is in fact just a risk premium that compensates for higher volatilities of less liquid stocks. By Proposition~\ref{thm:frictionless}, the frictionless equilibrium expected return is 
\begin{align*}
\bar{\mu}  = \bar{\gamma}\alpha\alpha^\top s,
\end{align*}
where $\alpha\alpha^\top$ is the frictionless equilibrium variance. We proxy $\mu$ and $\alpha\alpha^\top$ by the empirical estimates $\hat{\mu}$ and $\hat\Sigma$ reported above. The aggregate risk aversion $\bar\gamma$ is in turn estimated via a linear regression model without intercept as $\bar\gamma=2.97\times 10^{-13}$. Using the empirical covariance matrix and this calibrated value for the aggregate risk aversion $\bar\gamma$, the frictionless Black-Scholes return are $(7.76\%,7.55\%,7.56\%)^\top$. To wit, the (co-)variances of the high-, medium-, and low-liquidity portfolios observed empirically suggest nearly identical risk premia for all of them. This in stark contrast to the empirical data, where the low-liquidity portfolio has a substantially higher return than the portfolio composed of the highly liquid assets.

\subsection{Calibration of the Frictional Model}

We now discuss how the above calibration results change when trading costs (again proxied by ILLIQ) are taken into account. To ease the computational burden, we assume that the dividend volatility $\alpha\alpha^\top$ and in turn the leading-order equilibrium price volatilities are the same as in the frictionless version of the model. Likewise, we use the same value for the aggregate risk aversion $\bar{\gamma}$. Unlike in the frictionless version of the model, not just this aggregate risk aversion, but also the heterogeneity between the individual agents now play a crucial role. For tractability, we focus on the simplest model with two agents and write
$$\gamma^2=k\gamma^1,$$
where $k\ge1$ measures the heterogeneity of the two agents. Initially, we choose $k=2$ to illustrate the following calibration process; then, $\gamma^2=4.45\times 10^{-12}$ and $\gamma^1=8.91\times 10^{-13}$. However, by virtue of our explicit asymptotic formulas, different values of $k$ will just lead to a rescaling of the leading-order equilibrium returns implied by the model, which we outline at the end of this section. 

For simplicity, we assume that the transaction costs matrix $\Lambda$ is diagonal, which is reasonable if the quadratic trading costs are seen as a more tractable proxy for proportional costs. The diagonal elements of $\Lambda$ are the transaction costs for three portfolios proxied by ILLIQ as described above, multiplied by 9. This multiplication makes the transaction costs for the three portfolios comparable to a model with a single risky asset (with three times the order flow and whence nine times the quadratic costs). In particular, our estimate $\Lambda=\diag\{0.1269,0.3354,0.8595\}\times 10^{-8}$ is of the same order of magnitude as the direct estimates obtained from a proprietary database of trades in~\cite{collin2020liquidity}.\footnote{Estimating the transaction costs using the Bachelier volatilities divided by volume as implied by Kyle's model~\cite{kyle.85} gives comparable results: $(0.0785,0.2507,0.3137)^\top\times 10^{-8}$.} 

To complete the model specification, it now remains to estimate the endowment volatilities $\xi$. This is difficult, since these are not observable. As a way out, we extend the approach developed in~\cite{gonon.al.19} for a single risky asset and calibrate these parameters to time series data for trading volume. To this end, recall from~\eqref{eq: asymp phidot} that, at the leading order for small costs, the (signed) trading volume $\dot\varphi$ approximately has the Ornstein-Uhlenbeck dynamics
\begin{align*}
d\dot{\bar{\varphi}}_t&=-\kappa_1\dot{\bar{\varphi}}_tdt + \kappa_2dW_t, ~~
\text{where }\kappa_1=\sqrt{\frac{\gamma^1+\gamma^2}{2}} \Lambda^{-1}\left(\Lambda\#\alpha\alpha^\top\right),~\kappa_2=-\kappa_1\cdot \left(\alpha \alpha^\top\right)^{-1} \alpha \xi.
\end{align*}
Since $\kappa_1$ is positive definite, the stationary distribution of $\dot\varphi_t$ has the density~\cite[Section 6.5]{risken1991fokker}
\begin{align*}
p(\mathbf{x})=(2\pi)^{-D/2}	(\det\Omega)^{-1/2}\exp\left(-\frac{1}{2}\mathbf{x}^\top\Omega^{-1}\mathbf{x}\right),
\end{align*}
where $\Omega$ satisfies the algebraic Riccati equation
\begin{align*}
\kappa_1\Omega+\Omega\kappa_1^\top=\kappa_2\kappa_2^\top.
\end{align*}
By the ergodic theorem and the explicit formula for absolute moments of Gaussian distribution~\cite{nabeya1951absolute}, it follows that the long-run averages averages of the second moments of the trading volumes have the following closed-form expression:
\begin{align}\label{eq: volume moments}
\begin{split}
&\lim_{T\to\infty}\frac{1}{T}\int_0^T|(\dot\varphi_t)_i(\dot\varphi_t)_j|dt=\int_\mathbb{R}|x_ix_j|p_{ij}(x_i,x_j)dx_idx_j\\
&\quad=\frac{2(\Omega_{ii}\Omega_{jj})^{1/2}}{\pi}\Gamma(1)^2\mathcal{H}(-1/2,-1/2,1/2,\rho_{ij}^2)\text{ for }i\neq j,~~\Omega_{ii}\text{ for }i=j.
\end{split}
\end{align}
(Here, $\rho_{ij}=\Omega_{ij}/(\Omega_{ii}\Omega_{jj})^{1/2}$ and $\mathcal{H}$ is Gauss' hypergeometric function.) We use this explicit formula to calibrate the (unobservable) volatility matrices $\xi$ of the agents' endowment streams as follows. We assume that this $3\times3$ matrix is symmetric, and use an initial guess (for our numerical results, we used $-I_3\times 10^9$)\footnote{Here, negative diagonal elements produce positive liquidity premia in line with the data.} to calculate $\kappa_2$, and in turn the left hand side of \eqref{eq: volume moments} for $i,j=1,2,3,i\le j$. We then compare the result with the second moments of daily trading volume observed empirically. The parameters of the matrix $\xi$ are in turn updated using the global optimizer \texttt{GlobalSearch} in MATLAB in order to find the parameter values that match the empirical data as well as possible. The result is
    \begin{align*}
    \xi=\begin{pmatrix}
   -2.07  &  1.91  &  0.64\\
    1.91  & -1.77  & -0.59\\
    0.64  & -0.59  & -0.20\\
    \end{pmatrix} \times 10^9.
    \end{align*}
This corresponds to the second moments of daily volumes
\begin{align*}
(5.64,0.75,0.33,1.89,0.49,1.26)^\top \times 10^{17},	
\end{align*}
which are very close to the second moments of the daily volumes observed in our dataset:
\begin{align*}
(5.63,0.71,0.32,1.92,0.46,1.28)^\top \times 10^{17}.
\end{align*}
   
  With all parameters of the model specified, we can now calculate the leading-order adjustments of the equilibrium expected returns of the portfolios H, M and L due to transaction costs. To wit, equation~\eqref{eq: asymp lp} shows that the annualized (absolute) changes compared to the frictionless version of the model are $(-0.2440,-0.0074,0.0758)^\top$. After dividing by the corresponding average prices, we obtain the following adjustments of the annualized relative (Black-Scholes) returns: $(-0.5374\%,-0.0150\%,0.1979\%)^\top$. As a consequence, the expected return of the most liquid portfolio is indeed reduced, whereas the expected returns of the low liquidity portfolio is increased. When the heterogeneity parameter is chosen (somewhat arbitrarily) as $k=2$, the difference between the return adjustments is $0.74\%$ annually, substantially smaller than the difference of $2.7\%$ observed empirically.  
  
  \begin{figure}[htbp]
\centering
\includegraphics[scale=0.5]{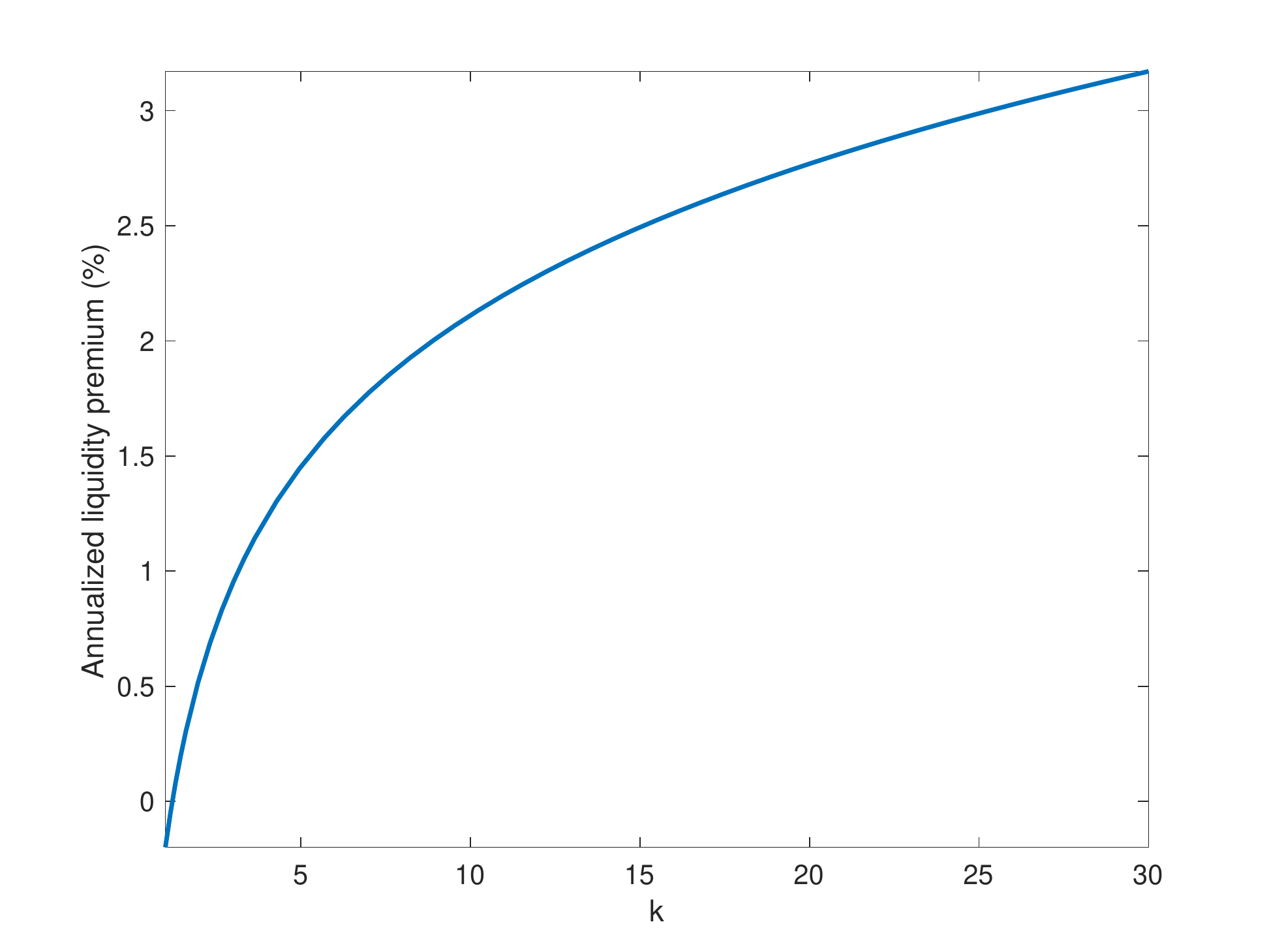}
\caption{Annualized liquidity premium (i.e., difference between the equilibrium relative returns of the L and H portfolios) plotted against the heterogeneity parameter $k$. The empirically observed liquidity premium is 2.69\%.} \label{fig:1}
\end{figure}
  
To study how this result depends on $k$, observe that~\eqref{eq: asymp lp} shows that the average return adjustments scale with $k$ by a factor of $(k-1)(k+1)^{-1/2}k^{-1/4}$. To see this, note that $\gamma^1=\bar{\gamma}(1+k)/k,\gamma^2=\bar{\gamma}(1+k)$, and thus $\kappa_1$ scales with $k$ by the factor $(1+k)k^{-1/2}$. By calibrating $\xi$ to match the same second moments of daily volumes, $\xi$ has a factor of $k^{1/4}(1+k)^{-1/2}$, and \eqref{eq: asymp lp} establishes the scaling factor of the average return adjustments. Therefore, increasing the heterogeneity $k$ increases the liquidity premium between the low and high liquidity portfolios. This is illustrated in Figure~\ref{fig:1}, which shows that to produce a realistic level of liquidity premia, our model requires a substantial level of heterogeneity in the agents' preferences. This corroborates the partial equilibrium literature on liquidity premia, which finds that additional features such as market closure~\cite{dai2016portfolio}, unobservable regime switches~\cite{chen2020incomplete}, or state-dependent transaction costs~\cite{acharya.pedersen.05,lynch.tan.11} are needed to reproduce realistic levels of liquidity premia. Incorporating these effects into a general equilibrium analysis is an important but challenging direction for future research.

\section{Proofs}\label{sec:proofs}

\subsection{Proofs for the Frictionless Version of the Model}

\begin{proof}[Proof of Proposition~\ref{thm:frictionless}]
It is readily verified that the proposed price process solves the BSDE system~\eqref{eq:BSDES}. The corresponding covariance matrix $\alpha\alpha^\top$ is invertible by assumption. Whence, each agent's individually optimal trading strategy is given by~\eqref{eq:strat}. In view of~\eqref{eq:eqnocosts}, this simplifies to
\begin{align}\label{eq:Strategiesnocosts}
\bar\varphi^n_t 
=\frac{\bar{\gamma}}{\gamma^n} s - \big(\alpha\alpha^\top\big)^{-1} \alpha\xi^n_t,   \qquad t \in [0,T].
\end{align}
In particular, these holdings are admissible because they are normally distributed. As the aggregate endowment is zero ($\sum_{n=1}^N\xi^n=0$), these strategies indeed sum to $s$ as required for market clearing.

For uniqueness, suppose there are two solutions with uniformly bounded volatilities. Then, both of these solve the BSDE with truncated (and hence globally Lipschitz) generator, and therefore coincide.
\end{proof}

\subsection{Analysis of the Riccati System}\label{proof of sec: frictional}

The crucial tool for the proof of our main result on the existence of equilibria with transaction costs is Lemma~\ref{thm:ivp}, which establishes wellposedness for the Riccati system~\eqref{eqs: Ricatti} characterizing this equilibrium. The proof of~Lemma~\ref{thm:ivp} is in turn based on a number of auxiliary estimates on matrix-valued ODEs that we develop first. 

We start with the properties of the risk-aversion matrix $\Gamma$ introduced in~\eqref{def: gamma}. Recall from~\eqref{ordered risk aversion} that, without loss of generality, agent $N$ is supposed to be the most risk-averse one.

\begin{lemma}\label{lem: Gamma psd}
The matrix $\Gamma$ is positive definite and has only positive eigenvalues.\footnote{In much of the literature, positive definite matrices are additionally required to be symmetric, because this is necessary to derive many useful properties. However, the matrix $\Gamma$ is generally not symmetric, and we in turn carry out the subsequent analysis without this convenient property. Notice that in the absence of symmetry, a square matrix with positive eigenvalues can fail to be positive definite, and a positive definite matrix can fail to have real eigenvalues.}
\end{lemma}

\begin{proof}
The second part of the assertion has been established in~\cite[Lemma A.5]{herdegen.al.20}. Therefore it remains to show that $b^\top \Gamma b > 0$ for all $b\in\R^M{\setminus \{0\}}$. Observe that $\Gamma$ is a ``diagonal minus rank-1'' matrix:
\begin{align*}
\Gamma = \diag\{\gamma^1, \cdots, \gamma^{N-1}\} - \frac{1}{N} \mathbbm{1}_{N-1}\mathbbm{1}_{N-1}^\top \diag\{\gamma^1 - \gamma^N, \cdots, \gamma^{N-1} - \gamma^N\}.\tag{\ref{def: gamma}}
\end{align*}
To show that this matrix is positive definite, we define 
\begin{align}
v := \diag\{\gamma^1, \cdots, \gamma^{N-1}\} \mathbbm{1}_{N-1} - \frac{1}{N-1} \sum_{n=1}^{N-1} \gamma^n \mathbbm{1}_{N-1},
\end{align}
and observe that $v$ and $\mathbbm{1}_{N-1}$ are orthogonal:
\begin{align*}
v^\top \mathbbm{1}_{N-1} = \sum_{n=1}^{N-1} \gamma^n - \frac{1}{N-1} \sum_{n=1}^{N-1} \gamma^n \mathbbm{1}_{N-1}^\top\mathbbm{1}_{N-1} = \sum_{n=1}^{N-1} \gamma^n -\sum_{n=1}^{N-1} \gamma^n = 0.
\end{align*}
Whence, every vector $b\in\R^{N-1}$ has an orthogonal decomposition, in that there exist unique $a_1, a_v\in\R$ and $b_\perp \in\R^{N-1}$, such that
\begin{align*}
b = a_1 \mathbbm{1}_{N-1} + a_v v + b_\perp, \qquad \mbox{where } \mathbbm{1}_{N-1}^\top b_\perp = 0 =v^\top b_\perp. 
\end{align*}
With this notation, a direct calculation yields 
\begin{align*}
b^\top  \diag\{\gamma^1, \cdots, \gamma^{N-1}\} b 
& = b^\top \left(  a_1 \diag\{\gamma^1, \cdots, \gamma^{N-1}\} \mathbbm{1}_{N-1} +  \diag\{\gamma^1, \cdots, \gamma^{N-1}\} \left(a_v v + b_\perp\right)\right)
\\& = a_1 b^\top \left( v + \frac{1}{N-1} \sum_{n=1}^{N-1} \gamma^n \mathbbm{1}_{N-1}\right) + b^\top \diag\{\gamma^1, \cdots, \gamma^{N-1}\} \left(a_v v + b_\perp\right)
\\&=a_1^2\sum_{n=1}^{N-1} \gamma^n + 2a_1 a_v \|v\|^2 +\left(a_v v + b_\perp\right)^\top \diag\{\gamma^1, \cdots, \gamma^{N-1}\} \left(a_v v + b_\perp\right)
\\&\geq a_1^2\sum_{n=1}^{N-1} \gamma^n + 2a_1 a_v \|v\|^2 .
\end{align*}
(Here, we have used $\mathbbm{1}_{N-1}^\top \diag\{\gamma^1, \cdots, \gamma^{N-1}\}=(v+ \frac{1}{N-1} \sum_{n=1}^{N-1} \gamma^n \mathbbm{1}_{N-1})^\top$ in the second to last step.) Similarly, we can calculate
\begin{align*}
& \frac{1}{N} b^\top \mathbbm{1}_{N-1}\mathbbm{1}_{N-1}^\top \diag\{\gamma^1 - \gamma^N, \cdots, \gamma^{N-1} - \gamma^N\} b
 \\&\quad=  \frac{1}{N} b^\top \mathbbm{1}_{N-1}\mathbbm{1}_{N-1}^\top \diag\{\gamma^1 , \cdots, \gamma^{N-1} \} b - \frac{\gamma^N}{N} b^\top \mathbbm{1}_{N-1}\mathbbm{1}_{N-1}^\top b
\\ &\quad = \frac{1}{N} a_1 (N-1) \left(v + \frac{1}{N-1} \sum_{n=1}^{N-1} \gamma^n \mathbbm{1}_{N-1}\right)^\top b - a_1^2\frac{(N-1)^2}{N} \gamma^N
 \\&\quad = \frac{N-1}{N} a_1 \left(a_v \|v\|^2 + a_1\sum_{n=1}^{N-1} \gamma^n  \right) - {a_1^2}\frac{(N-1)^2}{N} \gamma^N
 \\&\quad = \frac{N-1}{N} \left(a_1 a_v \|v\|^2 + a_1^2 \sum_{n=1}^{N-1} (\gamma^n - \gamma^N) \right).
\end{align*}
As $ \diag\{\gamma^1, \cdots, \gamma^{N-1}\}$ is positive definite and $N \geq 2$, these two identities lead to the estimate
\begin{align*}
b^\top \Gamma b
&= b^\top  \diag\{\gamma^1, \cdots, \gamma^{N-1}\} b - \frac{1}{N} b^\top \mathbbm{1}_{N-1}\mathbbm{1}_{N-1}^\top \diag\{\gamma^1 - \gamma^N, \cdots, \gamma^{N-1} - \gamma^N\} b
\\&{>} \frac{N-1}{2N} b^\top  \diag\{\gamma^1, \cdots, \gamma^{N-1}\} b -  \frac{1}{N} b^\top \mathbbm{1}_{N-1}\mathbbm{1}_{N-1}^\top \diag\{\gamma^1 - \gamma^N, \cdots, \gamma^{N-1} - \gamma^N\} b
\\&\geq \frac{N-1}{2N} \left( a_1^2\sum_{n=1}^{N-1} \gamma^n + 2a_1 a_v \|v\|^2 \right)- \frac{N-1}{N} \left(a_1 a_v \|v\|^2 + a_1^2 \sum_{n=1}^{N-1} (\gamma^n - \gamma^N) \right)
\\& \geq  \frac{N-1}{2N}  a_1^2 \left(\sum_{n=1}^{N-1} \gamma^n + 2 \sum_{n=1}^{N-1} (\gamma^N - \gamma^n) \right)>0,
\end{align*}
where we have taken into account~\eqref{ordered risk aversion} in the last step. Whence, $\Gamma$ is indeed positive definite. 
\end{proof}

For later use,
we recall the definition of the operator norm, in which we will express our estimates for matrix ODEs below:
\begin{definition}\label{def:opnomr}
The \emph{operator norm} of an $M_1 \times M_2$ matrix $A$ is defined by 
\begin{align*}
\|A\|_{\mathrm{op}} :=\sup\{ \|Ab\| : b\in\R^{M_2}, \|b\| = 1\} .
\end{align*}
\end{definition}

\begin{remark}\label{properties of operator norm}
For the convenience of the reader, let us summarize the properties of the operator norm and the Frobenius norm from~\cite[Chapter 5]{horn2012matrix} and the properties of Kronecker product from~\cite[Chapter 2]{steeb1997matrix}. that we will use repeatedly and without further mention below: 
\begin{enumerate}
\item $\|A\|_{\mathrm{op}} = \|A^\top\|_{\mathrm{op}} = \frac{1}{2} \|A +A^\top\|_{\mathrm{op}}$. 
\item The operator norm is submultiplicative in that, for an $M_1 \times M_2$ matrix $A$ and an $M_2\times M_3$ matrix $B$, 
$$
\|AB\|_{\mathrm{op}} \leq \|A\|_{\mathrm{op}} \|B\|_{\mathrm{op}}. 
$$
\item For an $M_1 \times M_2$ matrix $A$, the corresponding operator norm and Frobenius norm are related by
$$
\|A\|_{\mathrm{op}}\leq \|A\| \leq \sqrt{M_1+M_2} \|A\|_{\mathrm{op}}. 
$$
\item For the Kronecker product of two matrices (of arbitrary dimension), we have
$$
\|A\otimes B\|_{\mathrm{op}} = \|A\|_{\mathrm{op}}\|B\|_{\mathrm{op}}.
$$
\item The transpose of the Kronecker product satisfies: 
$$(A\otimes B)^\top = A^\top \otimes B^\top.$$
\item Bilinearity and associativity of Kronecker products:
\begin{align*}
A\otimes (B+C) &= A\otimes B +A\otimes C,\\
(B+C)\otimes A & = B\otimes A + C\otimes A, \\
A\otimes 0 &= 0\otimes A = 0. 
\end{align*}
\item The mixed-product property of Kronecker products: for matrices $A$, $B$, $C$ and $D$ of appropriate dimensions,
$$
(A\otimes B)(C\otimes D) = AC \otimes BD. 
$$
\end{enumerate}
\end{remark}

We now verify that the Kronecker product preserves positive-semidefiniteness as long as its second argument is also symmetric:

\begin{lemma}\label{otimes psd}
If matrices $A$, $B$ are positive semidefinite and $B$ is symmetric, then the Kronecker product $A\otimes B$ is also positive semidefinite. 
\end{lemma}

\begin{proof}
Notice that $A + A^\top$ is symmetric positive semidefinite. Thus, $A + A^\top$ and $B$ are both diagonalizable, in that there exist orthogonal matrices $P$, $Q$ and diagonal matrices $D_A$, $D_B$ such that 
\begin{align*}
P D_A P^\top = A+A^\top , \qquad 
Q D_B Q^\top = B = B^\top. 
\end{align*}
Here, the diagonal elements of $D_A$ and $D_B$ are the eigenvalues of $A+A^\top$ and $B$, respectively. These are all nonnegative because these matrices are both positive semidefinite \emph{and} symmetric.
As a consequence, the Kronecker product $D_A \otimes D_B$ is also diagonal with nonnegative diagonal elements; in particular, it is also positive semidefinite. It follows that $A\otimes B + \left( A\otimes B\right)^\top $ is also positive semidefinite, because the symmetry of B and the properties of the Kronecker product allow us to rewrite this matrix as
\begin{align*}
 A\otimes B + \left( A\otimes B\right)^\top 
  &= A\otimes B+A^\top \otimes B^\top
\\&= A\otimes B+A^\top \otimes B 
\\&= \left(A+A^\top\right) \otimes B 
\\&= \left(P D_A P^\top\right) \otimes \left(Q D_B Q^\top\right)
\\&=\left(P \otimes Q\right) \left(D_A \otimes D_B\right) \left(P^\top \otimes Q^\top\right)
\\&=\left(P \otimes Q\right) \left(D_A \otimes D_B\right) \left(P \otimes Q\right) ^\top,
\end{align*}
and $P$, $Q$ are orthogonal matrices. Whence, the matrix $A\otimes B$ is also positive semidefinite. 
\end{proof}

With this toolbox, we now establish some properties of \emph{linear} matrix ODEs that will be used below to bound the Riccati system~\eqref{eqs: Ricatti}.

\begin{lemma}\label{psd}
{Let $A: \mathbb{R}_+ \to \mathbb{R}^{M \times M}$ be a continuous function with $A(0)=0$}, and let $Y$ be the unique solution~\cite[Theorem 2.4, Definition 2.12]{chicone2006ordinary} 
of the linear matrix ODE
\begin{equation}\label{eq:matrix eq}
Y'(\tau) = A(\tau)Y(\tau), \quad Y (0) = I_{M}. 
\end{equation}
Suppose that $Y''(\tau) = B(\tau)Y(\tau)$, where $B(\tau)$ is positive semidefinite for all $\tau\geq 0$. Then the matrix $A(\tau)$ is positive semidefinite for all $\tau\geq 0$ as well. 
\end{lemma}

\begin{proof}
Differentiation and the ODE~\eqref{eq:matrix eq} give
\begin{align}\label{first derivative}
\left(Y^\top(\tau) Y(\tau)\right)'
&=  \left(Y'(\tau)\right)^\top Y(\tau) + Y^\top (\tau) Y'(\tau)
=Y^\top(\tau) \left(A^\top(\tau) +A(\tau)\right) Y(\tau)
\end{align}
and, in turn,
\begin{align*}
  \left(Y^\top(\tau) Y(\tau)\right)''
  &= \left(Y''(\tau)\right)^\top Y(\tau)  + 2 \left(Y'(\tau)\right)^\top Y'(\tau)+ Y^\top (\tau) Y''(\tau)
\\&= Y^\top(\tau) \left(B^\top (\tau) + B(\tau) + 2 A^\top (\tau) A(\tau) \right) Y(\tau). 
\end{align*}
For every $b\in\R^M$,  we thus have
\begin{align}\label{second derivative}
  \left(b^\top Y^\top(\tau) Y(\tau) b \right)''
  &= b^\top Y^\top(\tau) \left(B^\top (\tau) + B(\tau) + 2 A^\top (\tau) A(\tau) \right) Y(\tau)b\notag
\\&= \left(Y(\tau)b \right)^\top \left(B^\top (\tau) + B(\tau) + 2 A^\top (\tau) A(\tau) \right) Y(\tau)b \geq 0,
\end{align}
because $B(\tau)$, $B^\top(\tau)$ and $A^\top (\tau) A(\tau)$ are all positive semidefinite. Thus $\tau \mapsto \left(b^\top Y^\top(\tau) Y(\tau) b\right)' $ is increasing on $\mathbb{R}_+$ and~\eqref{first derivative} in turn yields 
\begin{align}
2 b^\top Y^\top(\tau) A(\tau) Y(\tau)b 
&= b^\top Y^\top(\tau) \left(A^\top(\tau) +A(\tau)\right) Y(\tau) b \notag\\
& = \left(b^\top Y^\top(\tau) Y(\tau) b\right)' \notag\\
&\geq \left(b^\top Y^\top(0) Y(0) b\right)' \notag\\
&= b^\top Y^\top(0) \left(A^\top(0) +A(0)\right) Y(0) b=0. \label{eq:mon}
\end{align}
By Liouville's formula~\cite[Proposition 2.18]{chicone2006ordinary}, $Y(\tau)$ is invertible for every $\tau\geq 0$. Hence, for every $b\in\R^M$,
\begin{align*}
b^\top A(\tau) b = (Y^{-1}(\tau)b)^\top Y^\top(\tau) A(\tau) Y(\tau) Y^{-1}(\tau) b \geq 0.
\end{align*}
$A(\tau)$ therefore is indeed positive semidefinite for every $\tau\geq 0$. 
\end{proof}

\begin{lemma}\label{operator norm}
Let $Y$ be the unique solution 
of the linear matrix ODE~\eqref{eq:matrix eq}. If $\tau \mapsto A(\tau)$ is continuous  and $A(\tau)$ is positive semidefinite for every $\tau\geq 0$, then $\tau \mapsto \|Y(\tau)\|_{\mathrm{op}}$ is increasing.
\end{lemma}

\begin{proof}
For $b\in\R^M$ with $\|b\| = 1$ and $\tau\geq r\geq 0$,~\eqref{first derivative} and~\eqref{eq:mon} imply 
\begin{align*}
\|Y(\tau)b\|^2 = b^\top Y^\top(\tau) Y(\tau) b \geq b^\top Y^\top(r) Y(r) b = \|Y(r) b \|^2\geq 0.
\end{align*}
As a consequence, the operator norm of $Y(\tau)$ is indeed increasing in $\tau$:
\begin{align*}
\|Y(\tau)\|_{\mathrm{op}} = \sup \{\|Y(\tau)b\|: \|b\| = 1\} \geq \sup \{\|Y(r)b\|: \|b\| = 1\}  = \|Y(r)\|_{\mathrm{op}}. 
\end{align*}
\end{proof}

\begin{corollary}\label{F psd}
Let $(F,H)$ be the unique \emph{local} solution of the Riccati system~\eqref{eqs: Ricatti} on its maximal interval of existence $[0,T_{\max})$.  Then $F(\tau)$ is positive semidefinite for every $\tau\in[0,T_{\max})$. 
\end{corollary}

\begin{proof}
First, recall that $\Lambda$ and $\Lambda^{1/2}$ are both symmetric and positive definite, and hence also invertible. 
Let $\Phi_F$ be the solution (on $[0,T_{\max})$) of the linear matrix ODE
\begin{equation}
\Phi_F'(\tau) = \left(I_{N-1}\otimes \Lambda^{-1/2}\right) F^\top(\tau) \left(I_{N-1}\otimes \Lambda^{-1/2}\right) \Phi_F(\tau), \quad  \Phi_F (0) = I_{K(N-1)}.
\label{eq:matrix}
\end{equation}
Differentiation of this matrix function, the linear ODE~\eqref{eq:matrix} for $\Phi_F$, and the Riccati equation~\eqref{eqs: Ricatti} for $F$ imply
\begin{align*}
\Phi_F''
&=  \left(I_{N-1}\otimes \Lambda^{-1/2}\right)\left( F'\right)^\top \left(I_{N-1}\otimes \Lambda^{-1/2}\right)\Phi_F+  \left(I_{N-1}\otimes \Lambda^{-1/2}\right) F^\top \left(I_{N-1}\otimes \Lambda^{-1/2}\right) \Phi_F'\\
&= \left(I_{N-1}\otimes \Lambda^{-1/2}\right)\left( F' + F \left(I_{N-1}\otimes \Lambda^{-1}\right) F\right)^\top \left(I_{N-1}\otimes \Lambda^{-1/2}\right) \Phi_F\\
&= \left(I_{N-1}\otimes \Lambda^{-1/2}\right)\left( \Gamma\otimes \left(\alpha+ \left({c} \otimes I_K\right)^\top H \right)\left(\alpha+ \left({c} \otimes I_K\right)^\top H \right)^\top\right)^\top \left(I_{N-1}\otimes \Lambda^{-1/2}\right)\Phi_F.
\end{align*}
The matrix $\Gamma\otimes [(\alpha+ (c \otimes I_K)^\top H(\tau))(\alpha+({c} \otimes I_K)^\top H(\tau))^\top]$ is positive semidefinite by Lemmas~\ref{otimes psd} and~\ref{lem: Gamma psd}. As $\Lambda$ and in turn also $I_{N-1}\otimes \Lambda^{-1/2}$ are symmetric and positive definite, it follows that $\left(I_{N-1}\otimes \Lambda^{-1/2}\right)\left(\Gamma\otimes [(\alpha+ (c \otimes I_K)^\top H(\tau))(\alpha+({c} \otimes I_K)^\top H(\tau))^\top]\right)\left(I_{N-1}\otimes \Lambda^{-1/2}\right)^\top$ is also positive semi-definite for every $\tau\in[0,T_{\text{max}})$.  Together with Lemma~\ref{psd}, it follows that the matrix 
\begin{equation}
\left(I_{N-1}\otimes \Lambda^{-1/2}\right) F^\top(\tau) \left(I_{N-1}\otimes \Lambda^{-1/2}\right) \quad \mbox{is positive semidefinite} \label{eq:psd}
\end{equation}
for every $\tau\in[0,T_{\text{max}})$ as well. The assertion now follows from~\eqref{eq:psd} and the identity
\begin{align*}
 F^\top(\tau) 
 &= \left(I_{N-1}\otimes \Lambda^{1/2}\right) \left(I_{N-1}\otimes \Lambda^{-1/2}\right) F^\top(\tau) \left(I_{N-1}\otimes \Lambda^{-1/2}\right)\left(I_{N-1}\otimes \Lambda^{1/2}\right)
\\&=   \left(I_{N-1}\otimes \Lambda^{1/2}\right) \left(I_{N-1}\otimes \Lambda^{-1/2}\right) F^\top(\tau) \left(I_{N-1}\otimes \Lambda^{-1/2}\right)\left(I_{N-1}\otimes \Lambda^{1/2}\right)^\top. 
\end{align*}
\end{proof}

\begin{corollary}\label{Psi norm}
With the solution $\Phi_F$ of the linear matrix ODE~\eqref{eq:matrix}, define 
\begin{align}\label{eq:state}
\Psi_F (r;\tau) = \Phi_F(r)\Phi_F^{-1}(\tau), \qquad r, \tau \in[0,T_{\max}{)}.
\end{align}
 Then $\|\Psi_F (r;\tau) \|_{\mathrm{op}} \leq1$ for every $0\leq r \leq \tau < T_{\max}$. 
\end{corollary}
\begin{proof}
By the ODE for $\Phi_F(r)$, we have 
\begin{align*}
\frac{\partial}{\partial r} \Psi_F(r;\tau) &= \left(I_{N-1}\otimes \Lambda^{-1/2}\right) F^\top(r) \left(I_{N-1}\otimes \Lambda^{-1/2}\right) \Psi_F(r;\tau).
\end{align*}
In view of Lemma~\ref{otimes psd} and~\eqref{eq:psd}, $\left(I_{N-1}\otimes \Lambda^{-1/2}\right) F^\top(r) \left(I_{N-1}\otimes \Lambda^{-1/2}\right) $ is positive semidefinite for every $r\in[0,T_{\max})$. 
Lemma~\ref{operator norm} in turn yields  
\begin{align*}
\| \Psi_F(r;\tau)\|_{\mathrm{op}} \leq \| \Psi_F(\tau;\tau)\|_{\mathrm{op}} = \| \Phi_F(\tau)\Phi_F^{-1}(\tau)\|_{\mathrm{op}} = 1,
\end{align*} 
for every $0\leq r \leq \tau\leq T_{\max}$, as asserted.
\end{proof}

After the above preparations, we now turn to the proof of Lemma~\ref{thm:ivp}.

\begin{proof}[Proof of Lemma~\ref{thm:ivp}]
We show that the local solution $(F,H)$ of the Riccati equation is in fact a global solution because it remains bounded on any finite time interval (so that $T_{\max}=\infty$). 

To this end, first observe that the ODEs~\eqref{eqs: Ricatti} for $F$ and~\eqref{eq:matrix} for {$\Phi_F$} give
\begin{align*}
\Big( \Phi_F^\top (I_{N-1}\otimes \Lambda^{-1/2})  F\Big)'
& = \left( \Phi_F'\right)^\top \left(I_{N-1}\otimes \Lambda^{-1/2}\right)  F +  \Phi_F^\top  \left(I_{N-1}\otimes \Lambda^{-1/2}\right)  F'\\
& =  \Phi_F^\top\left(I_{N-1}\otimes \Lambda^{-1/2}\right) \left(F\left(I_{N-1}\otimes \Lambda^{-1}\right)F +   F'\right)\\
& =  \Phi_F^\top\left(I_{N-1}\otimes \Lambda^{-1/2}\right) \Gamma\otimes \left[(\alpha+ ({c} \otimes I_K)^\top H )(\alpha+(c\otimes I_K)^\top H)^\top\right]\\
& = \Phi_F^\top \left(\Gamma\otimes \Lambda^{-1/2}\left(\alpha+ \left({c} \otimes I_K\right)^\top H \right)\left(\alpha+ \left({c}\otimes I_K\right)^\top H \right)^\top\right) .
\end{align*}
Together with $F(0)=0$, it follows that 
\begin{align}
&\quad  \Phi_F^\top(\tau)  \left(I_{N-1}\otimes \Lambda^{-1/2}\right)  F(\tau) \label{eq:prod}\\
 &= \int_0^\tau \Phi_F^\top(r)\left(I_{N-1}\otimes \Lambda^{-1/2}\right) \left(\Gamma\otimes \left(\alpha+ \left({c} \otimes I_K\right)^\top H(r) \right)\left(\alpha+ \left({c} \otimes I_K\right)^\top H (r)\right)^\top\right) \; dr. \notag
\end{align}
By Liouville's Formula~\cite[Proposition 2.18]{chicone2006ordinary}, we have
\begin{align*}
\det\left( \Phi_F(\tau)\right)
= \exp\left(\int_0^\tau \tr\left(\left(I_{N-1}\otimes \Lambda^{-1/2}\right) F^\top(\tau) \left(I_{N-1}\otimes \Lambda^{-1/2}\right)\right) dr \right)\det\left( \Phi_F(0)\right)>0, 
\end{align*}
so that $\Phi_F(\tau)$ is invertible for all $\tau<T_{\max}$. We can in turn solve~\eqref{eq:prod} for $F$ by multiplying with the inverse of {${\Phi}_F$} and the inverse $I \otimes \Lambda^{1/2}$ of $I \otimes \Lambda^{-1/2}$. With the notation from~\eqref{eq:state}, this leads to
\begin{align}\label{eq: F}
&\quad F(\tau)=\notag\\
&  \int_0^\tau\left(I_{N-1}\otimes \Lambda^{1/2}\right)   \Psi_F^\top(r;\tau) \left(\Gamma\otimes 
\Lambda^{-1/2}\left(\alpha+ \left({c} \otimes I_K\right)^\top H \right)\left(\alpha+ \left({c} \otimes I_K\right)^\top H \right)^\top\right)(r) \; dr.
\end{align}
Similarly, after multiplying $\Phi_F^\top(\tau)  \left(I_{N-1}\otimes \Lambda^{-1/2}\right) $ to the left of $H$, integrating, and then taking into account the ODE~\eqref{eqs: Ricatti} for $H$, we obtain
\begin{align}\label{eq: H}
H(\tau) &= \int_0^\tau\left(I_{N-1}\otimes \Lambda^{1/2}\right)  \Psi_F^\top(r;\tau)\left(\Gamma\otimes \Lambda^{-1/2}\left(\alpha+ \left({c} \otimes I_K\right)^\top H(r) \right)\right)\xi \; dr. 
\end{align}
As a consequence,
\begin{align*}
&\quad \alpha+  \left({c} \otimes I_K\right)^\top H(\tau) 
\\&=\alpha + \left({c} \otimes I_K\right)^\top\left( \int_0^\tau\left(I_{N-1}\otimes \Lambda^{1/2}\right)   \Psi_F^\top(r;\tau)\left(\Gamma\otimes \Lambda^{-1/2}\left(\alpha+ \left({c} \otimes I_K\right)^\top H(r) \right)\right)\xi \; dr\right)
\\&= \alpha +  \int_0^\tau\left({c}\otimes \Lambda^{1/2}\right)   \Psi_F^\top(r;\tau)\left(\Gamma\otimes \Lambda^{-1/2}\left(\alpha+ \left({c} \otimes I_K\right)^\top H(r) \right)\right)\xi \; dr.
\end{align*}
Recalling the definition of operator norm of a matrix from Definition~\ref{def:opnomr} and the properties from  Remark~\ref{properties of operator norm}, it follows that
\begin{align*}
&\quad \left\| \alpha+  \left({c} \otimes I_K\right)^\top H(\tau)  \right\|_{\mathrm{op}}
\\&\leq \|\alpha\|_{\mathrm{op}} + \left\| \int_0^\tau\left({c}\otimes \Lambda^{1/2}\right)   \Psi_F^\top(r;\tau)\left(\Gamma\otimes \Lambda^{-1/2}\left(\alpha+ \left({c} \otimes I_K\right)^\top H(r) \right)\right)\xi \; dr \right\|_{\mathrm{op}}
\\&\leq \|\alpha\|_{\mathrm{op}} +\int_0^\tau \left\| \left({c}\otimes \Lambda^{1/2}\right)   \Psi_F^\top(r;\tau)\left(\Gamma\otimes \Lambda^{-1/2}\left(\alpha+ \left({c} \otimes I_K\right)^\top H(r) \right)\right)\xi \right\|_{\mathrm{op}} \; dr
\\&\leq \|\alpha\|_{\mathrm{op}} +\int_0^\tau \|{c}\| \| \Lambda^{1/2}\|_{\mathrm{op}} \|\Psi_F(r;\tau)\|_{\mathrm{op}} \|\Gamma\|_{\mathrm{op}}  \| \Lambda^{-1/2}\|_{\mathrm{op}} \left\|\alpha+ \left({c} \otimes I_K\right)^\top H(r) \right\|_{\mathrm{op}} \|\xi\|_{\mathrm{op}} \; dr
\\&\leq \|\alpha\| + \|{c}\|\|\Lambda\|^{1/2} \|\Lambda^{-1}\|^{1/2}\|\Gamma\| \|\xi\| \int_0^\tau  \left\|\alpha+ \left({c} \otimes I_K\right)^\top H(r) \right\|_{\mathrm{op}} \; dr.
\end{align*}
(Here, the last step uses the estimate for $\|\Psi_F(r;\tau)\|_{\mathrm{op}} \leq 1$ from Corollary~\ref{Psi norm}.)  Gr\"onwall's inequality applied to the \emph{scalar} function $\tau \mapsto \| \alpha+  \left({c} \otimes I_K\right)^\top H(\tau)\|_{\mathrm{op}}$ in turn yields
\begin{align}\label{estimation of frictional volatility}
\left\| \alpha+  \left({c} \otimes I_K\right)^\top H(\tau)  \right\|_{\mathrm{op}} \leq \|\alpha\| \exp \left( \|{c}\|\|\Lambda\|^{1/2} \|\Lambda^{-1}\|^{1/2}\|\Gamma\| \|\xi\|  \tau\right), \quad \tau \in [0,T_{\max}).
\end{align}
Together with~\eqref{eq: H}, the fact that the Frobenius norm of a matrix is dominated by a constant times the operator norm, and $D\geq K$,  it now follows that
\begin{align*}
\|H(\tau) \| &\leq \sqrt{(N-1)(D+K)}\|H(\tau) \|_{\mathrm{op}} \\
&\leq\sqrt{2(N-1)D} \left\| \int_0^\tau\left(I_{N-1}\otimes \Lambda^{1/2}\right)   \Psi_F^\top(r;\tau)\left(\Gamma\otimes \Lambda^{-1/2}\left(\alpha+ \left({c} \otimes I_K\right)^\top H(r) \right)\right)\xi \; dr\right\|_{\mathrm{op}} 
\\&\leq
\sqrt{2(N-1)D}\int_0^\tau  \| \Lambda^{1/2}\|_{\mathrm{op}} \|\Psi_F(r;\tau)\|_{\mathrm{op}} \|\Gamma\|_{\mathrm{op}}  \| \Lambda^{-1/2}\|_{\mathrm{op}} \left\|\alpha+ \left({c} \otimes I_K\right)^\top H(r) \right\|_{\mathrm{op}} \|\xi\|_{\mathrm{op}}  \; dr
\\&\leq \sqrt{2(N-1)D}\int_0^\tau\| \Lambda^{1/2}\|  \|\Gamma\|  \| \Lambda^{-1/2}\| \left\|\alpha+ \left({c} \otimes I_K\right)^\top H(r) \right\|_{\mathrm{op}} \|\xi\| \; dr
\\&\leq \sqrt{2(N-1)D}\int_0^\tau \|\alpha\| \|\Lambda\|^{1/2} \|\Lambda^{-1}\|^{1/2}\|\Gamma\| \|\xi\|  \exp \left( \|{c}\|\|\Lambda\|^{1/2} \|\Lambda^{-1}\|^{1/2}\|\Gamma\| \|\xi\|  r \right) dr.
\end{align*}
The integral in the upper bound is finite for any $\tau$. All components of $H$ therefore remain uniformly bounded on any finite time interval. Similarly, by~\eqref{eq: F},
\begin{align*}
\|F(\tau)\|
&\leq \sqrt{2(N-1)K}\int_0^\tau \|\Lambda\|^{1/2} \|\Lambda^{-1}\|^{1/2}\|\Gamma\|  \|\alpha\|^2 \exp \left( 2\|{c}\|\|\Lambda\|^{1/2} \|\Lambda^{-1}\|^{1/2}\|\Gamma\| \|\xi\|  r\right) dr.
\end{align*}
Whence, all elements of $F$ also remain uniformly bounded on any finite time interval. The local solution of the Riccati system~\eqref{eqs: Ricatti} therefore also is a global solution on any finite time interval.
\end{proof} 

\subsection{Existence of Equilibria with Transaction Costs}

Having established wellposedness for the Riccati system~\eqref{eqs: Ricatti}, we now turn to the proof of our main result on the global existence of equilibria with transaction costs.

\begin{proof}[Proof of Theorem~\ref{thm:Radnerfric}]
In view of Lemma~\ref{lem:FBSDE}, we have to verify that the candidate processes from Theorem~\ref{thm:Radnerfric} indeed solve the FBSDE system from Lemma~\ref{lem:FBSDE}.

First, recall that $F$ is positive definite by Corollary~\ref{F psd}. Lemma~\ref{operator norm} applied to the linear matrix Riccati equation~\eqref{eq:matrix exponential} in turn shows that $\Phi(t)$ is well-defined on $[0,T]$ and has increasing operator norm. Together with the proof of Lemma~\ref{thm:ivp}, it follows that the functions $F$, $H$, $\Phi$ are all uniformly bounded on $[0,T]$, and it is in turn straightforward to verify using the Gaussian law of the driving Brownian motion that the candidate solution $({\varphi}, \dot{\varphi}, \dot{Z}, \bar{S}+\mathcal{Y}-(c\otimes \Lambda)^\top \dot{\varphi}, \bar{\sigma}-(c\otimes \Lambda)^\top \dot{Z})$ indeed belongs to $\mathbb{H}^4(\R^{K(N-1)}) \times \mathbb{H}^4(\R^{K(N-1)}) \times \mathbb{H}^2(\R^{K(N-1) \times D}) \times \mathcal{S}^2(\R^K) \times \mathbb{H}^4(\R^{K \times D})$. Hence, it remains to verify that these processes also satisfy the dynamics and initial/terminal conditions from Lemma~\ref{lem:FBSDE}. 

To this end, recall that by Liouville's formula~\cite[Proposition 2.18]{chicone2006ordinary}, the matrix $\Phi(t)$ is invertible for each $t \in [0,T]$. Differentiation of this matrix function and the ODE~\eqref{eq:matrix exponential} give
\begin{align}\label{eq:drift1}
\left(\Phi^{-1}(t) \right)' 
&= - \Phi^{-1}(t) \Phi'(\tau) \Phi^{-1}(\tau) 
= -\Phi^{-1}(t) \left(I_{N-1}\otimes \Lambda^{-1/2}\right) F^\top(T-t) \left(I_{N-1}\otimes \Lambda^{-1/2}\right).
\end{align}
Moreover, by definition of the function $\Psi$ in~\eqref{eq:state transition},
$$
\Phi^\top (t) \left(I_{N-1}\otimes \Lambda^{1/2}\right)  \Psi^\top (r; t)= (\Psi(r;t) \left(I_{N-1}\otimes \Lambda^{1/2}\right) \Phi(t))^\top = \Phi^\top (r) \left(I_{N-1}\otimes \Lambda^{1/2}\right) . 
$$
With these observations, we can rewrite~\eqref{eq:frictional position} as 
\begin{align*}
&\Phi^\top (t)  \left(I_{N-1}\otimes \Lambda^{1/2}\right) \left(\varphi_t - \bar{\varphi}_0\right)\\
&\quad =  \left(I_{N-1}\otimes \Lambda^{1/2}\right) \left(\varphi_{0-} - \bar{\varphi}_0 \right)
- \int_0^t\Phi^\top(r) \left(I_{N-1}\otimes \Lambda^{-1/2}\right)H(T-r) W_r dr,
\end{align*}
so that
\begin{align}\label{eq:drift2}
d\; \Phi^\top (t)  \left(I_{N-1}\otimes \Lambda^{1/2}\right) \left(\varphi_t - \bar{\varphi}_0\right) = -\Phi^\top(t) \left(I_{N-1}\otimes \Lambda^{-1/2}\right)H(T-t) W_t dt.
\end{align}
Integration by parts and the dynamics~\eqref{eq:drift1}-\eqref{eq:drift2} in turn give
\begin{align}
d \varphi_t  
& = d \left[\left(I_{N-1}\otimes \Lambda^{-1/2}\right) \left(\Phi^\top (t)\right)^{-1} \Phi^\top (t)  \left(I_{N-1}\otimes \Lambda^{1/2}\right) \left(\varphi_t - \bar{\varphi}_0\right)\right] \notag\\
&=- \left(I_{N-1}\otimes \Lambda^{-1}\right)  F(T-t) \left(I_{N-1}\otimes \Lambda^{-1/2}\right) \left(\Phi^\top (t)\right)^{-1}\Phi^\top (t)  \left(I_{N-1}\otimes \Lambda^{1/2}\right) \left(\varphi_t - \bar{\varphi}_0\right) dt\notag\\
&\quad- \left(I_{N-1}\otimes \Lambda^{-1/2}\right) \left(\Phi^\top (t)\right)^{-1} \Phi^\top(t) \left(I_{N-1}\otimes \Lambda^{-1/2}\right)H(T-t) W_t dt \notag\\
& = -\left(I_{N-1}\otimes \Lambda^{-1}\right) \left[F(T-t) \left(\varphi_t - \bar{\varphi}_0\right) + H(T-t) W_t \right]dt \label{eq:drift3}\\
&= \dot{\varphi}_t dt.\notag
\end{align}
Moreover,
$$
\varphi_0 = \bar{\varphi}_0 + \Psi^\top(0;0)\left(\varphi_{0-} - \bar{\varphi}_0\right)  - 0 = \varphi_{0-}, 
$$
so that the first equation of the FBSDE system in Lemma~\ref{lem:FBSDE} is indeed satisfied. 

To verify that the other two equations from Lemma~\ref{lem:FBSDE} are satisfied as well, we first observe the following identities for the matrix $\Gamma$ from~\eqref{def: gamma} and the vector $c$ defined in Lemma~\ref{thm:ivp}:
\begin{align}
\Gamma \left(c + \frac{\bar{\gamma}}{\gamma^N}\mathbbm{1}_{N-1}\right) 
& = \frac{\gamma^N}{N} \mathbbm{1}_{N-1}, \label{eq:1}\\
\Gamma^\top c &= \frac{1}{N}\begin{bmatrix}
\gamma^N-\gamma^n &
\cdots &
\gamma^N-\gamma^{N-1}
\end{bmatrix}^\top, \label{eq:2}\\
\mathbbm{1}_{N-1}^\top c & = 1 - \bar{\gamma}\frac{N}{\gamma^N}. \label{eq:3}
\end{align}

With the (constant) frictionless equilibrium volatility $\bar{\sigma}=\alpha$ from Lemma~\ref{thm:frictionless} and the process $\dot{Z}$ from~\eqref{eq:volatility}, the candidate for the frictional equilibrium volatility is
\begin{align}\label{eq:frictional volatility}
\sigma_t =\bar{\sigma}_t-(c\otimes \Lambda)^\top \dot{Z}_t = \alpha + (c\otimes \Lambda)^\top \left(I_{N-1}\otimes \Lambda^{-1}\right)H(T-t) = \alpha +(c\otimes I_K)^\top H(T-t).
\end{align}
The definition of $\dot{\varphi}$ in \eqref{eq:frictional strategy}, integration by parts, the Riccati equations~\eqref{eqs: Ricatti} for $F$, $H$, and~\eqref{eq:frictional volatility} in turn lead to
\begin{align}
d \dot{\varphi}_t  
  & = -\left(I_{N-1}\otimes \Lambda^{-1}\right) d  \left[F(T-t) \left(\varphi_t - \bar{\varphi}_0\right) + H(T-t) W_t \right] \notag
\\& = \left(I_{N-1}\otimes \Lambda^{-1}\right) \left[\left(F'(T-t) \left(\varphi_t - \bar{\varphi}_0\right)  + H'(T-t) W_t\right)dt - F(T-t) d\left(\varphi_t - \bar{\varphi}_0\right) - H(T-t) dW_t \right] \notag
\\& = \left(I_{N-1}\otimes \Lambda^{-1}\right) \left(F'(T-t)+ F(T-t)\left(I_{N-1}\otimes \Lambda^{-1}F(t-t)\right)\right) \left(\varphi_t - \bar{\varphi}_0\right)  dt \notag
\\&\quad+ \left(I_{N-1}\otimes \Lambda^{-1}\right)\left(H'(T-t) +F(T-t)\left(I_{N-1}\otimes \Lambda^{-1}\right)H(T-t) \right) W_t dt  \notag
\\&\quad-\left(I_{N-1}\otimes \Lambda^{-1}\right)H(T-t) dW_t \notag
\\& = \left(I_{N-1}\otimes {\Lambda^{-1}} \right)\left( \left(\Gamma\otimes \sigma_t \sigma_t^\top\right) \left(\varphi_t - \bar{\varphi}_0\right) +  \left(\Gamma\otimes \sigma_t\right)\xi W_t\right)dt +\dot{Z}_t dW_t \notag
\\& = \left( \left(\Gamma\otimes \Lambda^{-1}\sigma_t \sigma_t^\top\right) \varphi_t +  \left(\Gamma\otimes {\Lambda^{-1}}\sigma_t\right)\xi_t- \left(\Gamma\otimes \Lambda^{-1}\sigma_t \sigma_t^\top\right)\bar{\varphi}_0\right)dt +\dot{Z}_t dW_t. \label{eq:driftdphi}
\end{align}
With the explicit form of $\bar{\varphi}^n_0$ from~\eqref{eq:Strategiesnocosts}, we can write $\bar{\varphi}_0 = (c + \frac{\bar{\gamma}}{\gamma^N}\mathbbm{1}_{N-1}) \otimes s$. Together with~\eqref{eq:1}, it follows that
\begin{align*}
\left(\Gamma\otimes \Lambda^{-1}\sigma_t \sigma_t^\top\right)\bar{\varphi}_0
  & = \Gamma \left(c + \frac{\bar{\gamma}}{\gamma^N}\mathbbm{1}_{N-1}\right) \otimes \Lambda^{-1}\sigma_t \sigma_t^\top s
 = \frac{\gamma^N}{N} \mathbbm{1}_{N-1}  \otimes \Lambda^{-1}\sigma_t \sigma_t^\top s.
\end{align*}
This shows that the candidate processes indeed match the dynamics in the second equation from Lemma~\ref{lem:FBSDE}. In view of the initial conditions of the Riccati system~\eqref{eqs: Ricatti}, the corresponding terminal condition is also satisfied:
$$
\dot{\varphi}_T = - \left(I_{N-1}\otimes \Lambda^{-1}\right) \left[F(0) \left(\varphi_T - \bar{\varphi}_0\right) + H(0) W_T\right] = 0.
$$
Finally, for the frictionless equilibrium price $\bar{S}$ from Lemma~\ref{thm:frictionless} and $\mathcal{Y}$ defined in~\eqref{eq:Ybis}, 
\begin{align*}
&d \left(\bar{S}_t +\mathcal{Y}_t\right) \\
&\quad=\bar{\gamma}  \left(\alpha\alpha^\top + \left({c} \otimes I_K\right)^\top H \alpha^\top + \alpha H^\top\left({c} \otimes I_K\right) +\left({c} \otimes I_K\right)^\top H H^\top  \left({c} \otimes I_K\right)\right)(T-t) s dt + \alpha dW_t
\\&\quad=\bar{\gamma}  \left(\alpha+ \left({c} \otimes I_K\right)^\top H(T-t) \right) \left(\alpha+ \left({c} \otimes I_K\right)^\top H(T-t) \right)^\top s dt + \alpha dW_t
\\& \quad=\bar{\gamma} \sigma_t\sigma_t^\top s dt + \alpha dW_t.
\end{align*}
(Here, we have used \eqref{eq:frictional volatility} in the last step.) Next, observe that the dynamics of $\dot{\varphi}$ computed in ~\eqref{eq:driftdphi} and the identities~\eqref{eq:2}, \eqref{eq:3} give 
\begin{align*}
(c\otimes \Lambda)^\top d \dot{\varphi}_t
& = \left( \left(c^\top\Gamma\otimes\sigma_t \sigma_t^\top\right) \varphi_t +  \left(c^\top\Gamma\otimes \sigma_t\right)\xi_t-  \frac{\gamma^N}{N} c^\top \mathbbm{1}_{N-1}  \otimes \sigma_t \sigma_t^\top s \right)dt +(c\otimes \Lambda)^\top \dot{Z}_t dW_t
\\&= \left(\frac{1}{N}\sigma_t  \sum_{n=1}^{N-1}(\gamma^N - \gamma^n) \left(\sigma_t^\top\varphi^n_t + \xi^n_t \right) + \left(\bar\gamma - \frac{\gamma^N}{N} \right)\sigma_t \sigma_t^\top s  \right) dt + (c\otimes \Lambda)^\top \dot{Z}_t dW_t.
\end{align*}
For $S_t = \bar{S}_t +\mathcal{Y}_t - (c\otimes \Lambda)^\top \dot{\varphi}_t$ from Theorem~\ref{thm:Radnerfric}, the dynamics we have just computed as well as~\eqref{eq:frictional volatility} show
\begin{align*}
dS_t 
&= \bar{\gamma} \sigma_t\sigma_t^\top s dt  -  \left(\frac{1}{N}\sigma_t  \sum_{n=1}^{N-1}(\gamma^N - \gamma^n) \left(\sigma_t^\top\varphi^n_t + \xi^n_t \right) + \left(\bar\gamma - \frac{\gamma^N}{N} \right)\sigma_t \sigma_t^\top s  \right) dt +\sigma_t dW_t
\\&= \left( \frac{\gamma^N}{N} \sigma_t \sigma_t^\top s +\frac{1}{N}\sigma_t  \sum_{n=1}^{N-1}(\gamma^n - \gamma^N) \left(\sigma_t^\top\varphi^n_t + \xi^n_t \right) \right) dt +\sigma_t dW_t.
\end{align*}
The third equation in Lemma~\ref{lem:FBSDE} is therefore also satisfied, because the corresponding terminal condition is matched as well:
$$
S_T = \bar{S}_T + \mathcal{Y}_T - (c\otimes \Lambda)^\top\dot{\varphi}_T  = \mathfrak{S} - 0 - 0 = \mathfrak{S}, 
$$
This completes the proof.
\end{proof}

\subsection{Proof of the Asymptotic Expansions}\label{proof of asymptotics}

The rigorous convergence proof for the asymptotics approximations is based on estimates for the largest and smallest singular values of the involved matrices. We first recall the definition and the properties of singular values of matrices. Then, we establish bounds on the singular values of the solutions of linear matrix ODEs in Lemma~\eqref{matrix exponential bounds}. Using this tool and a matrix version of the variation on of constants formula, we then derive estimates for the solution $F^\lambda, H^\lambda$ of the Riccati ODEs~\eqref{eqs: Ricatti} as a function of the asymptotic parameter $\lambda$. These in turn allow us to show that the functions can be approximated by constant matrices that solve some \emph{algebraic} Riccati equations. With these approximations at hand, we then proof the asymptotic expansions of the equilibrium price and trading volume from Theorem~\ref{thm:asymp}.

\begin{definition}\label{singular value and svd}
The \emph{singular values} of a real-valued $M_1 \times M_2$ matrix $A$ are the square roots of the non-negative eigenvalues of $A A^\top$. 
\end{definition}

For the convenience of the reader, we summarize the properties of singular values from~\cite[Chapter 2]{horn2012matrix} that we henceforth use without further mention. 

\begin{remark}
Let $A$ be a real-valued $M_1 \times M_2$ matrix. 
\begin{enumerate}
\item $A$ and $A^\top$ have the same non-zero singular values, but not the same as $\frac{1}{2}(A+A^\top)$. 
\item If $A$ is symmetric and $M_1 = M_2$, then the absolute value of the eigenvalues of $A$ coincide with the singular values.
\item Minimax representation for singular values:  Let $\sigma_{\max}(A)$  and $\sigma_{\min}(A)$ denote the largest and smallest singular value of $A$, respectively. Then,
$$
\sigma_{\max}(A) = \sup\{\|A^\top b\| : b \in \mathbb{R}^{M_1}, \|b\| = 1\}, \quad  \sigma_{\min}(A) = \inf\{ \|A^\top b\|: b \in \mathbb{R}^{M_1}, \|b\|=1\}. 
$$
In particular, $\|A\|_{\mathrm{op}}  = \sigma_{\max}(A)$.
\end{enumerate}
\end{remark}

We now consider the linear matrix ODE~\eqref{eq:matrix eq}, which is the key ingredient for the matrix version of the variation of constants formula that we use to prove our asymptotic expansions below. The following lemma shows that bounds on the singular values of the matrix function on the right-hand side of~\eqref{eq:matrix eq} are inherited by the largest and smallest singular values of the solution:

\begin{lemma}\label{matrix exponential bounds}
Let $Y$ be the unique solution of the linear matrix ODE~\eqref{eq:matrix eq}. Suppose $\tau \mapsto A(\tau)$ is continuous  and $A(\tau)$ is positive semidefinite for every $\tau\geq 0$, 
with $a_{\max}>a_{\min}>0$ such that for every  $\tau\in [0,T]$:
$$
a_{\max} \geq \frac{1}{2} \sigma_{\max} (A(\tau)+A^\top(\tau)) \geq \frac{1}{2}\sigma_{\min} (A(\tau)+A^\top(\tau)) \geq a_{\min} >0.
$$ 
Then for every $0\leq r\leq \tau\leq T$,  
\begin{align}\label{singular value threshold}
e^{-a_{\max}(\tau-r)}\leq 
\sigma_{\min}\left(Y(r)Y^{-1}(\tau)\right) \leq
\sigma_{\max} \left(Y(r)Y^{-1}(\tau)\right) \leq e^{-a_{\min}(\tau-r)}.
\end{align}
\end{lemma}

\begin{proof}
By Liouville's formula~\cite[Proposition 2.18]{chicone2006ordinary}, both $Y(r)$ and $Y(\tau)$ are invertible, 
hence for every $b\in \mathbb{R}^M\setminus\{0\}$, we have $\left\|Y(r)Y^{-1}(\tau) b\right\|>0$. By~\eqref{first derivative}, we then have
\begin{align*}
\frac{\partial}{\partial \tilde{r}} \left\|Y( \tilde{r})Y^{-1}(\tau) b\right\|^2
&=\frac{\partial }{\partial  \tilde{r}}\left(b^\top  (Y^{-1}(\tau))^\top  Y^\top ( \tilde{r}) Y( \tilde{r}) Y^{-1}(\tau) b \right)\\
&= \left(Y^{-1}(\tau) b\right)^\top (Y^\top ( \tilde{r}) Y( \tilde{r}) )'\left(Y^{-1}(\tau) b\right)\\
&=  \left(Y( \tilde{r})Y^{-1}(\tau) b\right)^\top \left(A( \tilde{r})+A^\top( \tilde{r})\right) \left(Y( \tilde{r})Y^{-1}(\tau) b\right)\\
&\geq 2 a_{\min} \left\|Y( \tilde{r})Y^{-1}(\tau) b \right\|^2.
\end{align*}
Now divide by $ \left\|Y(\tilde{r})Y^{-1}(\tau) b\right\|^2$ and integrate from $r$ to $\tau$, obtaining
\begin{align*}
 e^{2a_{\min} (\tau - r)} \left\|Y(r)Y^{-1}(\tau) b\right\|^2 \leq \left\|Y(\tau)Y^{-1}(\tau) b\right\|^2 =\|b\|^2.
\end{align*}
Hence, 
$
\left\|Y(r)Y^{-1}(\tau) b\right\| / \|b\| \leq  e^{-a_{\min} (\tau - r)},
$
and in turn $\sigma_{\max} (Y(r)Y^{-1}(\tau) ) \leq e^{-a_{\min} (\tau - r)}$. 

Similarly, for every $b\in \mathbb{R}^M\setminus\{0\}$, 
\begin{align*}
\frac{\partial}{\partial \tilde{r}} \left\|Y( \tilde{r})Y^{-1}(\tau) b\right\|^2
&\leq 2 a_{\max} \left\|Y( \tilde{r})Y^{-1}(\tau) b\right\|^2.
\end{align*}
Dividing by $ \left\|Y(\tilde{r})Y^{-1}(\tau) b\right\|^2$ and integrating from $r$ to $\tau$ in turn yields
\begin{align*}
 \|b\|^2 =  \left\|Y(\tau)Y^{-1}(\tau) b\right\|^2 \leq  e^{2a_{\max} (\tau - r)} \left\|Y(r)Y^{-1}(\tau) b\right\|^2.
\end{align*}
Hence,
$
\left\|Y(r)Y^{-1}(\tau) b\right\|/\|b\| \geq  e^{-a_{\max} (\tau - r)},
$
and thus $\sigma_{\min} (Y(r)Y^{-1}(\tau) ) \geq  e^{-a_{\max} (\tau - r)}$ as asserted.
\end{proof}

Using this lemma, we now approximate the solution to the Riccati system~\eqref{eqs: Ricatti}. Recall that the (normalized) transaction cost matrix $\bar\Lambda$ is symmetric and positive definite and the risk aversion matrix $\Gamma$ only has positive eigenvalues, so their square roots $\bar\Lambda^{1/2}$ and $\Gamma^{1/2}$ are well defined.\footnote{Recall that the square-root of $\Gamma$ is well defined by~\cite[Theorem 1.29]{higham2008functions} even though this matrix has positive eigenvalues that is generally only positive semidefinite but not symmetric, compare~\cite[Lemma A.5]{bouchard.al.18}.} Also note that $\bar\Lambda$, $\bar\Lambda^{-1}$, $\bar\Lambda^{1/2}$ and $\bar\Lambda^{-1/2}$ commute.

\begin{lemma}\label{bounds}
Let $(F^\lambda,H^\lambda)$ be the solution of the Riccati system~\eqref{eqs: Ricatti} for small transaction costs $\Lambda^\lambda=\lambda \bar\Lambda$. Define the constant matrix:
\begin{align}\label{longrun F}
\widehat{F}:= \Gamma^{1/2}\otimes\left(\bar\Lambda\#\alpha\alpha^\top\right),
\end{align}
and recall the definition of $M$ from Theorem~\ref{thm:asymp},
\begin{align}\label{longrun M}
M:=\left({c} ^\top \Gamma^{1/2}\otimes \bar{\Lambda}\left(\bar{\Lambda}\#\alpha\alpha^\top \right)^{-1}\alpha\right). 
\end{align}
Then, as $\lambda \downarrow 0$, the following estimates hold:
\begin{align}
\| F^\lambda(\tau)\|_{\mathrm{op}} &= O(\lambda^{1/2}), \; \tau\in[0,T],\qquad &\int_0^T \| F^{\lambda}(\tau) - \lambda^{1/2}\widehat{F}\|_{\mathrm{op}} d\tau = O(\lambda),\label{approximation F}\\
\|H^\lambda(\tau)\|_{\mathrm{op}} &= O(\lambda^{1/2}), \; \tau\in[0,T],\qquad &\int_0^T \| (c\otimes I_K)^\top H^{\lambda}(\tau) - \lambda^{1/2}M\xi\|_{\mathrm{op}} d\tau = O(\lambda). \label{approximation H}
\end{align}

\end{lemma}

\begin{proof}
The asserted bounds will be derived from a matrix version of the variation of constant formula below. Compared to the one-dimensional case treated in~\cite[Chapter 4]{thesis_shi}, this is complicated by the fact that the involved matrices generally do not commute. To overcome this difficulty, we introduce the unique solutions $\Phi_{F^\lambda}$ and $\Phi_{\widehat F}$ on $[0,T]$ of the following linear matrix ODEs: 
\begin{align}
\Phi_{F^\lambda}'(\tau) &= \frac{1}{\lambda}\left(I_{N-1}\otimes\bar\Lambda^{-1/2}\right) {F^\lambda}^\top(\tau)\left(I_{N-1}\otimes\bar\Lambda^{-1/2}\right) \Phi_{F^\lambda}(\tau), &\qquad \Phi_{F^\lambda}(0) = I_{K(N-1)},\label{Phi for F} \\
\Phi_{\widehat F}'(\tau) &= \frac{1}{\lambda^{1/2}}\left(I_{N-1}\otimes\bar\Lambda^{-1/2}\right) \widehat{F} \left(I_{N-1}\otimes\bar\Lambda^{-1/2}\right)\Phi_{\widehat F} (\tau), &\qquad \Phi_{\widehat F}(0)= I_{K(N-1)}. \label{Phi for hat F} 
\end{align}
Moreover, for $0\leq r\leq \tau \leq T$, define
\begin{align}\label{def: Psi}
\Psi_{F^\lambda}(r;\tau) &:= \Phi_{F^\lambda}(r)  \Phi_{F^\lambda}^{-1} (\tau), \qquad 
\Psi_{\widehat F}(r;\tau) := \Phi_{\widehat F}(r) \Phi_{\widehat F}^{-1} (\tau), 
\end{align}

The proof of the asymptotic expansions then proceeds along the following steps:

\begin{enumerate}
\item[\emph{Step 1:}] Show that for every $\tau\in[0,T]$, $\|F^\lambda(\tau)\|_{\mathrm{op}} = O(\lambda^{1/2})$. 
\item[\emph{Step 2:}] Show that for every $\tau\in[0,T]$, $\|(c\otimes I_K)^\top H^\lambda (\tau)\| \leq \sigma_{\min}(\alpha)\left(1 - e^{ -  \|{c}\|\|\bar\Lambda\|^{1/2} \|\bar\Lambda^{-1}\|^{1/2}\|\Gamma\| \|\xi\| T}\right)$.
\item[\emph{Step 3:}] Show that for every $0\leq r\leq \tau \leq T$, $\int_0^\tau \Psi_{F^\lambda}(r;\tau) dr  = O(\lambda^{1/2})$.
\item[\emph{Step 4:}] Show that for every $\tau\in[0,T]$, $\|H^\lambda(\tau)\|_{\mathrm{op}} = O(\lambda^{1/2})$.
\item[\emph{Step 5:}] Show that the approximations of $F^\lambda$ and $H^\lambda$ in \eqref{approximation F} and \eqref{approximation H} are valid at the asserted orders.
\end{enumerate}

\emph{Step 1:}
Notice that $\widehat{F}$ is the solution of the \emph{algebraic} Riccati equation
$$
\widehat F \left(I_{N-1}\otimes\bar{\Lambda}^{-1}\right) \widehat F = \Gamma \otimes \alpha \alpha^\top. 
$$
To simplify notation, set
$G^\lambda = \alpha(H^\lambda)^\top ({c} \otimes I_K) +  ({c} \otimes I_K)^\top H^\lambda( \alpha +({c} \otimes I_K)^\top H^\lambda)^\top$. 
The difference between the function $F^\lambda$ and the constant $\lambda^{1/2} \widehat{F}$ satisfies
\begin{align*}
\left(F^\lambda - \lambda^{1/2}\widehat{F}\right)'
&= \left(F^\lambda\right)'
=\Gamma\otimes \left(\alpha + \left({c} \otimes I_K\right)^\top H^\lambda \right)\left(\alpha + \left({c} \otimes I_K\right)^\top H^\lambda \right)^\top - \frac{F^\lambda\left(I_{N-1}\otimes \bar\Lambda^{-1}\right)F^\lambda}{\lambda}\\
& = \frac{1}{\lambda} \left(\lambda\widehat F \left(I_{N-1}\otimes\bar\Lambda^{-1}\right) \widehat F - F^\lambda\left(I_{N-1}\otimes \bar\Lambda^{-1}\right)F^\lambda \right)  + \Gamma\otimes G^\lambda\\
& = \frac{1}{\lambda} \left(\lambda^{1/2}\widehat F \left(I_{N-1}\otimes\bar\Lambda^{-1}\right) (\lambda^{1/2}\widehat F - F^\lambda) + (\lambda^{1/2}\widehat F - F^\lambda)\left(I_{N-1}\otimes \bar\Lambda^{-1}\right)F^\lambda \right) 
 + \Gamma\otimes G^\lambda. 
\end{align*}
We now want to apply a version of the variation of constant formula to obtain explicit estimates even though the matrices involved generally do not commute. To this end, multiply $\Phi_{F^\lambda}^\top $ and $\Phi_{\widehat F}$ on the left and right of $F^\lambda - \widehat F$, respectively. Then, taking derivatives and plugging in $\left(F^\lambda\right)'$ yields
\begin{align*}
\left(\Phi_{F^\lambda}^\top \left(I_{N-1}\otimes\bar\Lambda^{-1/2}\right) \left(F^\lambda - \lambda^{1/2}\widehat{F}\right)\left(I_{N-1}\otimes\bar\Lambda^{-1/2}\right)\Phi_{\widehat F}\right)'
& =  \Phi_{F^\lambda}^\top\left( \Gamma\otimes \bar{\Lambda}^{-1/2} G^\lambda\bar{\Lambda}^{-1/2}\right) \Phi_{\widehat F}.
\end{align*}
Now recall the initial condition $F^\lambda(0)=0$ and integrate both sides, obtaining
\begin{align*}
&\Phi_{F^\lambda}^\top(\tau)\left(I_{N-1}\otimes\bar\Lambda^{-1/2}\right) \left(F^\lambda(\tau) - \lambda^{1/2}\widehat{F}\right)\left(I_{N-1}\otimes\bar\Lambda^{-1/2}\right)\Phi_{\widehat F}(\tau) \\
& = \int_0^\tau \Phi_{F^\lambda}^\top (r)\left( \Gamma\otimes \bar\Lambda^{-1/2}G^\lambda(r)\bar\Lambda^{-1/2}\right) \Phi_{\widehat F} (r) dr - \lambda^{1/2}\left(I_{N-1}\otimes\bar\Lambda^{-1/2}\right)  \widehat{F}\left(I_{N-1}\otimes\bar\Lambda^{-1/2}\right).
\end{align*}
(Here, the arguments are dropped to ease notation.) By definition of $\Psi_{F^\lambda}(r;\tau) $ and $\Psi_{\widehat F}(r;\tau)$ in~\eqref{def: Psi}, we have 
\begin{align}\label{eq:1}
&F^\lambda(\tau) - \lambda^{1/2}\widehat{F}
\notag\\
&=   -\lambda^{1/2}\left(I_{N-1}\otimes\bar\Lambda^{1/2}\right) \Psi_{F^\lambda}^\top(0;\tau) \left(I_{N-1}\otimes\bar\Lambda^{-1/2}\right)  \widehat{F}\left(I_{N-1}\otimes\bar\Lambda^{-1/2}\right) \Psi_{\widehat F}(0;\tau)\left(I_{N-1}\otimes\bar\Lambda^{1/2}\right)\notag\\
&\quad  \int_0^\tau \left(I_{N-1}\otimes\bar\Lambda^{1/2}\right) \Psi_{F^\lambda}^\top(r;\tau)\left( \Gamma\otimes \bar\Lambda^{-1/2}G^\lambda (r) \bar\Lambda^{-1/2}\right) \Psi_{\widehat F} (r;\tau)\left(I_{N-1}\otimes\bar\Lambda^{1/2}\right)  dr.
\end{align}
With $C_0 :=  \|{c}\|\|\bar\Lambda\|^{1/2} \|\bar\Lambda^{-1}\|^{1/2}\|\Gamma\| \|\xi\|  = O(1)$, the representation~\eqref{estimation of frictional volatility} for $ \alpha+  \left({c} \otimes I_K\right)^\top H^\lambda$ implies that,  for every $\tau\in[0,T]$, 
\begin{align*}
 \left\| \alpha+  \left({c} \otimes I_K\right)^\top H^\lambda(\tau)  \right\|_{\mathrm{op}}
 \leq \|\alpha\| \exp \left( \|{c}\|\|\lambda\bar\Lambda\|^{1/2} \|\lambda \bar\Lambda^{-1}\|^{1/2}\|\Gamma\| \|\xi\| \tau \right)
= \|\alpha\| e^{C_0 \tau }\leq  \|\alpha\| e^{C_0 T}. 
\end{align*}
Recall from Corollary~\ref{Psi norm}, that 
$$\|\Psi_{F^\lambda}(r;\tau)\|_{\mathrm{op}}\leq 1.$$
To derive a similar bound for $\|\Psi_{\widehat F}(r;\tau)\|_{\mathrm{op}}$, notice that 
$$
\left(I_{N-1}\otimes\bar\Lambda^{-1/2}\right)  \widehat{F}\left(I_{N-1}\otimes\bar\Lambda^{-1/2}\right)  = \Gamma^{1/2} \otimes \left(\bar\Lambda^{-1/2} \alpha\alpha^\top \bar\Lambda^{-1/2} \right)^{1/2}, 
$$
where $\left(\bar\Lambda^{-1/2} \alpha\alpha^\top \bar\Lambda^{-1/2} \right)^{1/2}$ is a symmetric positive definite matrix. 
By Lemma~\ref{lem: Gamma psd}, the smallest eigenvalue of $\Gamma + \Gamma^\top$ is strictly positive, so
\begin{align*}
\widehat F_{\min} := \frac{1}{2}\sigma_{\min}(\Gamma + \Gamma^\top)\sigma_{\min} \left(\bar\Lambda^{-1/2} \alpha\alpha^\top \bar\Lambda^{-1/2} \right)^{1/2} >0.
\end{align*}
Lemma~\ref{matrix exponential bounds} therefore yields the following upper bound for $\Psi_{\widehat F}(r;\tau)$, valid for every $0\leq r\leq \tau\leq T$:
\begin{align*}
\|\Psi_{\widehat F}(r;\tau)\|_{\mathrm{op}} \leq e^{-\frac{\widehat F_{\min}}{\lambda^{1/2}} (\tau-r)}. 
\end{align*}
Moreover, with the help of~\eqref{estimation of frictional volatility}, direct calculation yields 
\begin{align*}
\left\| \Gamma\otimes \bar\Lambda^{-1/2}G^\lambda(r)\bar\Lambda^{-1/2}\right\|_{\mathrm{op}}
&\leq \|\Gamma\|\|\bar\Lambda^{-1}\| \left(\|\alpha\| + \left\| \alpha+  \left({c} \otimes I_K\right)^\top H(r)  \right\|_{\mathrm{op}} \right) \|(c\otimes I_K)^\top H^\lambda (r) \|_{\mathrm{op}}\\
&\leq 2e^{C_0 T}  \|\Gamma\|\|\bar\Lambda^{-1}\|\alpha\|  \|(c\otimes I_K)^\top H^\lambda (r) \|_{\mathrm{op}}. 
\end{align*}
After taking into account the above estimates,~\eqref{eq:1} leads to the following bound for the difference between the solution of the Riccati system and its constant approximation: 
\begin{align}\label{estimation of F}
\|F^\lambda(\tau) -\lambda^{1/2} \widehat F\|_{\mathrm{op}}
&\leq 2e^{C_0 T}  \|\bar\Lambda\|\|\bar\Lambda^{-1}\|\Gamma\|\|\alpha\| \int_0^\tau \|\Psi_{F^\lambda}(r;\tau)\|_{\mathrm{op}} \|\Psi_{\widehat F}(r;\tau)\|_{\mathrm{op}} \|(c\otimes I_K)^\top H^\lambda (r) \|_{\mathrm{op}}\; dr\notag
\\& +  \lambda^{1/2} \|\bar\Lambda\| \|\bar\Lambda^{-1}\| \|\widehat F\|  \|\Psi_{\widehat F}(0;\tau)\|_{\mathrm{op}}\|\Psi_{ F^\lambda}(0;\tau)\|_{\mathrm{op}} \notag
\\&\leq 2e^{C_0 T}\|\bar\Lambda\| \|\bar\Lambda^{-1}\|  \|\Gamma\|\|\alpha\|   \int_0^\tau e^{-\frac{\widehat F_{\min}}{\lambda^{1/2}} (\tau-r)} \|(c\otimes I_K)^\top H^\lambda (r) \|_{\mathrm{op}} dr 
\notag\\& 
+  \lambda^{1/2} \|\bar\Lambda\| \|\bar\Lambda^{-1}\| \|\widehat F\|  \|\Psi_{\widehat F}(0;\tau)\|_{\mathrm{op}}\|\Psi_{ F^\lambda}(0;\tau)\|_{\mathrm{op}}.
\end{align}
Recalling~\eqref{eq: H} and~\eqref{estimation of frictional volatility}, and taking into account that $\|\Psi_{F^\lambda} (r;\tau)\|_{\mathrm{op}}\leq 1$ for all $0\leq r\leq \tau \leq T$, we have the following bound for $(c\otimes I_K)^\top H^\lambda(\tau)$ for small costs $\bar\Lambda^\lambda = \lambda \bar\Lambda$:
\begin{align*}
\|(c\otimes I_K)^\top H^\lambda (\tau) \|_{\mathrm{op}} 
&\leq  \int_0^\tau \|c \| \|\lambda\bar\Lambda\|^{1/2} \|\lambda\bar\Lambda^{-1}\|^{1/2}\|\Gamma\| \|\xi\| \|\alpha\| e^{C_0 r}  dr\\
&\leq \int_0^\tau \|c \| \|\bar\Lambda\|^{1/2} \|\bar\Lambda^{-1}\|^{1/2}\|\Gamma\| \|\xi\| \|\alpha\| e^{C_0 r} dr\\
&\leq \|\alpha\| e^{C_0 \tau}\leq \|\alpha\| e^{C_0 T}. 
\end{align*} 
The triangle inequality and~\eqref{estimation of F} in turn give
\begin{align*}
\|F^\lambda(\tau)\|_{\mathrm{op}} 
&\leq \lambda^{1/2} \|\widehat F\|_{\mathrm{op}}+\|F^\lambda(\tau) - \lambda^{1/2} \widehat F\|_{\mathrm{op}} \\\
&\leq  \lambda^{1/2} \|\widehat F\| + \lambda^{1/2} \|\bar\Lambda\| \|\bar\Lambda^{-1}\| \|\widehat F\|  \|\Psi_{\widehat F}(0;\tau)\|_{\mathrm{op}} + 2e^{2C_0 T}\|\bar\Lambda\| \|\bar\Lambda^{-1}\|  \|\Gamma\|\|\alpha\|^2 \int_0^\tau e^{-\frac{\widehat F_{\min}}{\lambda^{1/2}} (\tau-r)}dr\\
&\leq \lambda^{1/2} \left( \|\widehat F\|  +\|\bar\Lambda\| \|\bar\Lambda^{-1}\| \|\widehat F\| + \frac{2e \|\bar\Lambda\| \|\bar\Lambda^{-1}\|  \|\Gamma\|\|\alpha\|^2 }{\widehat F_{\min}}\right)\\
&= : \lambda^{1/2} F_{\max} = O(\lambda^{1/2}). 
\end{align*}
This completes Step 1. For later use, also note that this estimate implies 
$$
\frac{1}{2}\sigma_{\max}\left(\left(I_{N-1}\otimes\bar\Lambda^{-1/2}\right)\left(F^\lambda (\tau) +{F^\lambda}^\top(\tau)\right)\left(I_{N-1}\otimes\bar\Lambda^{-1/2}\right)\right) 
\leq \lambda^{1/2} \|\bar\Lambda^{-1}\|_{\mathrm{op}}  F_{\max}.
$$
By Lemma~\ref{matrix exponential bounds}, it follows that the smallest singular value of $\Psi_{F^\lambda}(r;\tau)$ for $\tau_0\leq r \leq \tau\leq \tau_1$ satisfies
\begin{align}\label{Psi min}
\sigma_{\min} (\Psi_{F^\lambda}(r;\tau) ) \geq e^{-\frac{F_{\max}\|\bar\Lambda^{-1}\|_{\mathrm{op}} }{\lambda^{1/2}}(\tau-r)}.
\end{align}

\emph{Step 2:}
Recall $C_0 =   \|{c}\|\|\bar\Lambda\|^{1/2} \|\bar\Lambda^{-1}\|^{1/2}\|\Gamma\| \|\xi\|  = O(1)$.  We prove the claim by contradiction. To this end, suppose there exists a time $\tau_2 \in[0,T]$ such that 
$$
\|(c\otimes I_K)^\top H^\lambda (\tau_2)\| > \sigma_{\min}(\alpha^\top)(1 - e^{-C_0 T}).
$$
Notice that the Frobenius norm $\tau \mapsto \|(c\otimes I_K)^\top H^\lambda (\tau)\|$ is continuous, and $\|(c\otimes I_K)^\top H^\lambda (0)\|=0$. By the intermediate value theorem, there exists $\tau_1\in[0,\tau_2]$ (depending on $\lambda$) such that 
\begin{equation}\label{eq:tau1}
\|(c\otimes I_K)^\top H^\lambda (\tau_1)\| = \sigma_{\min}(\alpha^\top)(1- e^{-C_0 T})
\end{equation}
and
$$
\|(c\otimes I_K)^\top H^\lambda (\tau)\|< \sigma_{\min}(\alpha^\top)(1- e^{-C_0 T}), \quad \mbox{for $\tau \in[0,\tau_1)$.}
$$
On $[0,\tau_1]$, we then have 
\begin{align}\label{new volatility}
\sigma_{\min} \left(  \alpha + \left({c} \otimes I_K\right)^\top H^\lambda(\tau)\right) 
&\geq \sigma_{\min}(\alpha) -  \|(c\otimes I_K)^\top H^\lambda (\tau) \|_{\mathrm{op}}\notag\\
&\geq \sigma_{\min}(\alpha) \left(1 - (1- e^{-C_0 T})\right) = \sigma_{\min}(\alpha) e^{-C_0 T}. 
\end{align}
Define $\tau_0 = \lambda^{1/2} /C_0$. Then for $\tau\in[0,\tau_0]$, the representation~\eqref{eq: H} for $H^\lambda$ yields 
\begin{align*}
\|(c\otimes I_K)^\top H^\lambda (\tau) \| &\leq \sqrt{K+D} \|(c\otimes I_K)^\top H^\lambda (\tau) \|_{\mathrm{op}} \\
&\leq \sqrt{K+D} \|\alpha\| \int_0^{\tau}\|{c}\|\|\bar\Lambda\|^{1/2} \|\bar\Lambda^{-1}\|^{1/2}\|\Gamma\| \|\xi\|e^{\|{c}\|\|\bar\Lambda\|^{1/2} \|\bar\Lambda^{-1}\|^{1/2}\|\Gamma\| \|\xi\| r} dr\\
&\leq \sqrt{K+D} \|\alpha\| \left(e^{C_0 \tau} -1\right)\\
&\leq \sqrt{K+D} \|\alpha\| \left(e^{C_0 \tau_0} -1\right)\\
& = \sqrt{K+D} \|\alpha\| \left(e^{\lambda^{1/2}} -1\right) = O(\lambda^{1/2}). 
\end{align*}
For sufficiently small $\lambda$, we thus have $\tau_0 < \tau_1$. 

We now derive an upper bound of $\|(c\otimes I_K)^\top H^\lambda (\tau)\|$ on $[\tau_0, \tau_1]$ that will lead to the desired contradiction to~\eqref{eq:tau1}. To this end, we first develop some upper and lower bounds for  $\Psi_{F^\lambda}(r;\tau)$ and $F^{\lambda}$. By the identity~\eqref{second derivative} and the initial condition $F^\lambda(0) = 0$, 
\begin{align*}
&\left(I_{N-1}\otimes\bar\Lambda^{-1/2}\right)\left(F^\lambda (\tau) +{F^\lambda}^\top(\tau)\right)\left(I_{N-1}\otimes\bar\Lambda^{-1/2}\right)
\\&= \int_0^\tau \Psi_{F^\lambda} ^\top (r;\tau) \left((\Gamma+\Gamma^\top)\otimes ( \alpha + \left({c} \otimes I_K\right)^\top H^\lambda(r)) ( \alpha + \left({c} \otimes I_K\right)^\top H^\lambda(r))^\top\right) \Psi_{F^\lambda} (r;\tau) dr \\
&\quad +2 \int_0^\tau \Psi_{F^\lambda} ^\top (r;\tau) \left(I_{N-1}\otimes\bar\Lambda^{-1/2}\right) F^{\lambda}(r) \left(I_{N-1}\otimes\bar\Lambda^{-1}\right) {F^\lambda}^\top (r) \left(I_{N-1}\otimes\bar\Lambda^{-1/2}\right) \Psi_{F^\lambda} (r;\tau) dr.
\end{align*}
Together with~\eqref{Psi min} and~\eqref{new volatility}, it follows that
\begin{align*}
&\sigma_{\min}\left(\left(I_{N-1}\otimes\bar\Lambda^{-1/2}\right)\left(F^\lambda (\tau) +{F^\lambda}^\top(\tau)\right)\left(I_{N-1}\otimes\bar\Lambda^{-1/2}\right)\right)\\
&\geq  \int_0^\tau \sigma_{\min}(\Gamma + \Gamma^\top) \sigma_{\min}(\bar\Lambda^{-1}) \sigma_{\min}^2( \alpha + \left({c} \otimes I_K\right)^\top H^\lambda(r)) \sigma_{\min}^2(\Psi_{F^\lambda}(r;\tau)  ) dr \\
&\geq \int_0^\tau \sigma_{\min}(\Gamma + \Gamma^\top) \sigma_{\min}(\bar\Lambda^{-1}) \sigma_{\min}^2( \alpha) e^{-2C_0 T} e^{-\frac{2F_{\max}\|\bar\Lambda^{-1}\|_{\mathrm{op}} }{\lambda^{1/2}}(\tau-r)} dr\\
&=  \lambda^{1/2}\frac{2\sigma_{\min}(\Gamma + \Gamma^\top) \sigma_{\min}(\bar\Lambda^{-1}) \sigma_{\min}^2( \alpha)}{ e^{2C_0 T}F_{\max}\|\bar\Lambda^{-1}\|_{\mathrm{op}}} \left(1 - e^{-\frac{2F_{\max}\|\bar\Lambda^{-1}\|_{\mathrm{op}} }{\lambda^{1/2}}\tau} \right)\\
&\geq  \lambda^{1/2} \frac{2\sigma_{\min}(\Gamma + \Gamma^\top) \sigma_{\min}(\bar\Lambda^{-1}) \sigma_{\min}^2( \alpha)}{ e^{2C_0 T}F_{\max}\|\bar\Lambda^{-1}\|_{\mathrm{op}}} \left(1 - e^{-\frac{2F_{\max}\|\bar\Lambda^{-1}\|_{\mathrm{op}} }{\lambda^{1/2}}\tau_0} \right)
\\& = \lambda^{1/2}  \frac{2\sigma_{\min}(\Gamma + \Gamma^\top) \sigma_{\min}(\bar\Lambda^{-1}) \sigma_{\min}^2( \alpha)}{ e^{2C_0 T}F_{\max}\|\bar\Lambda^{-1}\|_{\mathrm{op}}} \left(1 - e^{-\frac{2F_{\max}\|\bar\Lambda^{-1}\|_{\mathrm{op}} }{C_0}} \right): = 2\lambda^{1/2} F_{\min} =O(\lambda^{1/2}).
\end{align*}
Again by Lemma~\ref{matrix exponential bounds}, we can estimate the largest singular value of $\Psi_{F^\lambda}(r;\tau)$ for every $\tau_0\leq r \leq \tau\leq \tau_1$ as follows:
\begin{align}\label{Psi max}
\|\Psi_{F^\lambda}(r;\tau)\|_{\mathrm{op}}  = \sigma_{\max} \left(\Psi_{F^\lambda}(r;\tau)\right) 
\leq e^{-\frac{ F_{\min}}{\lambda^{1/2}}(\tau-r)}.
\end{align}

Therefore, after plugging in~\eqref{eq: H} and~\eqref{estimation of frictional volatility}, we can estimate the Frobenius norm of $(c\otimes I_K)^\top H^\lambda (\tau_1)$ as 
\begin{align*}
\|(c\otimes I_K)^\top H^\lambda (\tau_1)\| 
&\leq  \sqrt{K+D } \int_0^{\tau_1}\|c\|\|\Gamma\| \| \bar\Lambda^{1/2}\|\| \bar\Lambda^{-1/2}\|\|\xi\| \|\Psi_{F^\lambda}(r;\tau_1)\|_{\mathrm{op}}  \|\alpha +c\otimes I_K)^\top H^\lambda(r) \|_{\mathrm{op}} dr 
\\&\leq \sqrt{K+D }\|c\|\|\alpha \| \|\Gamma\| \| \bar\Lambda^{1/2}\|\| \bar\Lambda^{-1/2}\|\|\xi\| e^{C_0 T}\int_0^{\tau_1} \|\Psi_{F^\lambda}(r;\tau_1)\|_{\mathrm{op}}   dr 
\\&\leq \sqrt{K+D }\|c\|\|\alpha \| \|\Gamma\| \| \bar\Lambda^{1/2}\|\| \bar\Lambda^{-1/2}\|\|\xi\| e^{C_0 T}\left(\int_0^{\tau_0} 1   dr + \int_{\tau_0}^{\tau_1}   e^{-\frac{ F_{\min}}{\lambda^{1/2}}(\tau_1-r)}  dr\right)
\\&\leq \sqrt{K+D }\|c\|\|\alpha \| \|\Gamma\| \| \bar\Lambda^{1/2}\|\| \bar\Lambda^{-1/2}\|\|\xi\| e^{C_0 T}\left(\frac{\lambda^{1/2}}{C_0}+ \frac{\lambda^{1/2}}{F_{\min}} e^{-\frac{ F_{\min}}{\lambda^{1/2}}(\tau_1-r)}  \right)
\\& = O(\lambda^{1/2}).
\end{align*}
For sufficiently small $\lambda$,  this contradicts~\eqref{eq:tau1} and therefore completes the proof of Step 2.\\

\emph{Step 3:}
From Step 2, we know that the estimate~\eqref{Psi max} holds for $\lambda^{1/2} /C_0 = \tau_0 \leq r \leq \tau \leq T$. This upper bound in turn implies
$$
\int_{\tau_0}^{\tau} \|\Psi_{F^\lambda}(r;\tau)\|_{\mathrm{op}}  dr \leq \int_{\tau_0}^\tau  e^{-\frac{F_{\min}}{\lambda^{1/2}}(\tau-r)} dr = \frac{\lambda^{1/2}}{F_{\min}} \left(1 - e^{-\frac{F_{\min}}{\lambda^{1/2}}(\tau-\tau_0)}\right) \leq \frac{\lambda^{1/2}}{F_{\min}}. 
$$
Together with the coarser upper bound $\|\Psi_{F^\lambda}(r;\tau)\|_{\mathrm{op}} \leq 1$ (for $0\leq r\leq \tau \leq T$), the desired estimate now follows:
$$
\int_{0}^{\tau} \|\Psi_{F^\lambda}(r;\tau)\|_{\mathrm{op}}  dr = \int_0^{\tau_0}  \|\Psi_{F^\lambda}(r;\tau)\|_{\mathrm{op}}  dr + \int_{\tau_0}^\tau  \|\Psi_{F^\lambda}(r;\tau)\|_{\mathrm{op}}  dr 
\leq \tau_0 + \frac{\lambda^{1/2}}{F_{\min}} = \lambda^{1/2} \left(\frac{1}{C_0} + \frac{1}{F_{\min}}\right).
$$

\emph{Step 4. }
The representation~\eqref{eq: H} for $H^\lambda$, the estimate~\eqref{estimation of frictional volatility} and the bound for the integral of $\Psi_{F^\lambda}(r;\tau)$ from Step 3, lead to the following upper bound for the operator norm of $H^\lambda$:
\begin{align*}
\|H^{\lambda} (\tau)\|_{\mathrm{op}} 
&\leq \int_0^{\tau}\|\Gamma\| \| \bar\Lambda^{1/2}\|\| \bar\Lambda^{-1/2}\|\|\xi\| \|\Psi_{F^\lambda}(r;\tau)\|_{\mathrm{op}}  \|\alpha +c\otimes I_K)^\top H^\lambda(r) \|_{\mathrm{op}} dr \\
&\leq \|\alpha \| \|\Gamma\| \| \bar\Lambda^{1/2}\|\| \bar\Lambda^{-1/2}\|\|\xi\| e^{C_0 T}\left( \int_{0}^{\tau} \|\Psi_{F^\lambda}(r;\tau)\|_{\mathrm{op}}   dr \right)
= O(\lambda^{1/2}). 
\end{align*}

\emph{Step 5:} With the estimate from Steps 1-4, we can now complete the proof of Lemma~\ref{bounds}. For the approximation of $F^\lambda$, insert the bounds for $H^\lambda$ from Step 4 into~\eqref{estimation of F}, obtaining
\begin{align*}
 \| F^\lambda(\tau)  - \lambda^{1/2} \widehat F\| _{\mathrm{op}} 
& \leq  \lambda^{1/2} \|\bar\Lambda\| \|\bar\Lambda^{-1}\| \|\widehat F\|  \|\Psi_{\widehat F}(0;\tau)\|_{\mathrm{op}}\|\Psi_{ F^\lambda}(0;\tau)\|_{\mathrm{op}}
\notag\\& \quad
+ 2e^{C_0 T}\|\bar\Lambda\| \|\bar\Lambda^{-1}\|  \|\Gamma\|\|\alpha\|   \int_0^\tau e^{-\frac{\widehat F_{\min}}{\lambda^{1/2}} (\tau-r)} \|c\| \|H^\lambda (r)  \|_{\mathrm{op}} dr 
\\&= \lambda^{1/2} \|\bar\Lambda\| \|\bar\Lambda^{-1}\| \|\widehat F\|  \|\Psi_{\widehat F}(0;\tau)\|_{\mathrm{op}}\|\Psi_{ F^\lambda}(0;\tau)\|_{\mathrm{op}} + O(\lambda).
\end{align*}
Now, recall that $\|\Psi_{ F^\lambda}(0;\tau)\|_{\mathrm{op}}\leq 1$; integrating~\eqref{estimation of F} in turn yields the desired approximation of $F^\lambda$: 
\begin{align*}
\int_0^T \| F^\lambda(\tau)  - \lambda^{1/2} \widehat F\| _{\mathrm{op}}  d\tau 
&=  \lambda^{1/2} \|\bar\Lambda\| \|\bar\Lambda^{-1}\| \|\widehat F\|  \int_0^T \|\Psi_{\widehat F}(0;\tau)\|_{\mathrm{op}} d\tau  + O(\lambda)
\\&\leq \lambda^{1/2} \|\bar\Lambda\| \|\bar\Lambda^{-1}\| \|\widehat F\|  \frac{\lambda^{1/2}}{\widehat F_{\min}}  + O(\lambda) = O(\lambda).
\end{align*}
To derive an analogous result for $H^\lambda$, define
$$
\widehat H := \left(\Gamma^{1/2}\otimes \bar\Lambda \left(\bar\Lambda\# \alpha \alpha^\top \right)^{-1} \alpha\right)\xi.
$$
Observe that $\widehat{H}$ is the solution of the linear algebraic equation $\widehat F (I_{N-1}\otimes \bar\Lambda^{-1}) \widehat H  = \Gamma\otimes \alpha$. Whence we can express the difference between $\lambda^{1/2} \widehat H$ and the solution $H^\lambda (\tau)$ of the linear Riccati equation~\eqref{eqs: Ricatti} as
\begin{align*}
&\left( H^\lambda  - \lambda^{1/2} \widehat H \right)' =  \left(H^\lambda\right)' 
= \left( \Gamma \otimes \left(\alpha + (c\otimes I_K) ^\top H^\lambda \right)\right)\xi - \frac{1}{\lambda}F^\lambda \left(I_{N-1}\otimes \bar\Lambda ^{-1}\right) H^\lambda 
\\&=  \left(\Gamma \otimes (c\otimes I_K) ^\top H^\lambda\right)\xi +   \frac{1}{\lambda}  \left(\lambda\widehat F\left(I_{N-1}\otimes \bar\Lambda ^{-1}\right) \widehat H  -   F^\lambda
\left(I_{N-1}\otimes \bar\Lambda ^{-1}\right) H^\lambda   \right) 
\\& = \left(\Gamma \otimes (c\otimes I_K) ^\top H^\lambda\right)\xi +   \frac{1}{\lambda} F^\lambda\left(I_{N-1}\otimes \bar\Lambda ^{-1}\right) \left( \lambda^{1/2} \widehat H - H^\lambda   \right)
+\frac{1}{\lambda^{1/2}} \left(\lambda^{1/2} \widehat F  - F^\lambda \right) \left(I_{N-1}\otimes \bar\Lambda ^{-1}\right) \widehat H. 
\end{align*}
Similarly as above, a matrix version of variation of constants now yields
\begin{align*}
H^\lambda (\tau) - \lambda^{1/2} \widehat H 
& =- \lambda^{1/2} \Psi_{F^\lambda}^\top (0;\tau)  \widehat H + \int_{0}^\tau    \Psi_{F^\lambda}^\top (r;\tau)  \left(\Gamma \otimes (c\otimes I_K) ^\top H^\lambda(r)\right)\xi  dr \\
& \quad +\frac{1}{\lambda^{1/2}} \int_{0}^\tau    \Psi_{F^\lambda}^\top (r;\tau) \left(\lambda^{1/2} \widehat F  - F^\lambda(r) \right) \left(I_{N-1}\otimes \bar\Lambda ^{-1}\right) \widehat H dr. 
\end{align*}
The first term is of order $O(\lambda^{1/2}) \left\|   \Psi_{F^\lambda} (0;\tau) \right\|_{\mathrm{op}}=O(\lambda)$. The estimates from Step 3 and 4, and a direct calculation in turn show that the second term is of order $O(\lambda)$ as well:
\begin{align*}
 \left\|\int_{0}^\tau    \Psi_{F^\lambda}^\top (r;\tau)  \left(\Gamma \otimes (c\otimes I_K) ^\top H^\lambda(r)\right)\xi  dr\right\|_{\mathrm{op}}
&\leq \|c\| \|\Gamma\| \|\xi\| \int_{0}^{\tau} \|H^{\lambda} (r)\|_{\mathrm{op}}  \|\Psi_{F^\lambda}(r;\tau)\|_{\mathrm{op}}  dr  
= O(\lambda).
\end{align*}
Finally, for the third term in the above estimate for $H^\lambda  - \lambda^{1/2} \widehat H$, we have
\begin{align*}
& \left\|\int_{0}^\tau    \Psi_{F^\lambda}^\top (r;\tau) \left(\lambda^{1/2} \widehat F  - F^\lambda(r) \right) \left(I_{N-1}\otimes \bar\Lambda ^{-1}\right) \widehat H dr\right\|_{\mathrm{op}}
\\&\leq \int_{0}^\tau  \left\|   \Psi_{F^\lambda} (r;\tau) \right\|_{\mathrm{op}} \|\bar\Lambda^{-1}\| \|\widehat H\| \| \lambda^{1/2} \widehat F  - F^\lambda(r)\|_{\mathrm{op}} dr
\\&\leq \int_{0}^\tau  \left\|   \Psi_{F^\lambda} (r;\tau) \right\|_{\mathrm{op}} \|\bar\Lambda^{-1}\| \|\widehat H\|\left( \lambda^{1/2} \|\bar\Lambda\| \|\bar\Lambda^{-1}\| \|\widehat F\|  \|\Psi_{\widehat F}(0;r)\|_{\mathrm{op}}\|\Psi_{ F^\lambda}(0;r)\|_{\mathrm{op}} + O(\lambda)\right)dr
\\& =O( \lambda^{1/2})\int_{0}^\tau  \left\|   \Psi_{F^\lambda} (r;\tau) \right\|_{\mathrm{op}}  \|\Psi_{\widehat F}(0;r)\|_{\mathrm{op}} \|\Psi_{ F^\lambda}(0;r)\|_{\mathrm{op}} dr + O(\lambda) \int_{0}^\tau  \left\|   \Psi_{F^\lambda} (r;\tau) \right\|_{\mathrm{op}} dr
\\& = O(\lambda^{1/2} )  \left\|   \Psi_{F^\lambda} (0;\tau ) \right\|_{\mathrm{op}}  \int_{0}^\tau  \|\Psi_{\widehat F}(0;r)\|_{\mathrm{op}} dr +O(\lambda^{3/2}) . 
\\&=  O(\lambda )  \left\|   \Psi_{F^\lambda} (0;\tau ) \right\|_{\mathrm{op}} +O(\lambda^{3/2}) . 
\end{align*}
Together with the estimate from Step 3,  it follows that
$$
\int_0^T \| H^\lambda(\tau)  - \lambda^{1/2} \widehat H \|_{\mathrm{op}} d\tau  \leq O(\lambda) + O(\lambda^{1/2})  \int_{0}^T  \left\|   \Psi_{F^\lambda} (0;\tau) \right\|_{\mathrm{op}} d\tau =O(\lambda).
$$
Therefore, the assertion follows after recalling that $M \xi = (c\otimes I_K)^\top \widehat H$ by definition.
\end{proof} 

With the above approximations of the Riccati system~\eqref{eqs: Ricatti} at hand, we can now carry out the rigorous convergence proof for the asymptotic expansions from Theorem~\ref{thm:asymp}.

\begin{proof}[Proof of Theorem~\ref{thm:asymp}]
From~\eqref{eq:volatility} in Theorem~\ref{thm:asymp}, we have $\sigma_t - \bar{\sigma}_t = (c\otimes I_K ) H^\lambda(T-t)$. Hence the  approximation~\eqref{eq: small cost delta sigma} of the volatility correction due to small transaction costs follows directly from~\eqref{approximation H}. 

Next, we turn to the trading rate $\dot\varphi$. To this end, we first need a further estimation.
Notice that 
(the arguments are dropped here to ease notation)
\begin{align*}
& \left\|\left(\alpha+ \left({c} \otimes I_K\right)^\top H^\lambda \right) -  \left(\alpha+ \left({c} \otimes I_K\right)^\top H^\lambda \right)\left(\alpha+ \left({c} \otimes I_K\right)^\top H^\lambda\right)^\top \left(\alpha\alpha^\top\right)^{-1}\alpha\right\|
\\&=\left\|\left(\alpha+ \left({c} \otimes I_K\right)^\top H^\lambda \right)I_D -  \left(\alpha+ \left({c} \otimes I_K\right)^\top H^\lambda \right)\left(\alpha+ \left({c} \otimes I_K\right)^\top H^\lambda\right)^\top \left(\alpha\alpha^\top\right)^{-1}\alpha\right\|
\\&=\left\|\left(\alpha+ \left({c} \otimes I_K\right)^\top H^\lambda\right) \left( I_D - \left(\alpha+ \left({c} \otimes I_K\right)^\top H^\lambda\right)^\top \left( \alpha\alpha^\top\right)^{-1}\alpha  \right)\right\|
\\&=\left\|\left({c} \otimes I_K\right)^\top H^\lambda \left( I_D - \left(\alpha+ \left({c} \otimes I_K\right)^\top H^\lambda \right)^\top \left( \alpha\alpha^\top\right)^{-1}\alpha  \right) - \alpha \left(H^\lambda\right)^\top (c\otimes I_K)\left(\alpha\alpha^\top\right)^{-1} \alpha\right\|
\\&\leq \|c\| \|H^\lambda\| \left(\left\|  I_D - \left(\alpha+ \left({c} \otimes I_K\right)^\top H^\lambda \right)^\top \left( \alpha\alpha^\top\right)^{-1}\alpha \right\| + \|\alpha\| \left\|\left(\alpha\alpha^\top\right)^{-1} \alpha\right\| \right) = O(\lambda^{1/2}). 
\end{align*}
Define
$$
E^\lambda (\tau) =H^\lambda(\tau) - F^\lambda(\tau)\left(I_{N-1}\otimes\left( \alpha\alpha^\top\right)^{-1}\alpha\right)\xi.
$$
By the representations~\eqref{eq: F} and~\eqref{eq: H} for $F^\lambda$ and $H^\lambda$, we have 
\begin{align*}
 \left\|E^\lambda(\tau)\right\|_{\mathrm{op}} 
\leq O(\lambda^{1/2})\int_0^\tau\|\bar\Lambda\| \|\bar\Lambda^{-1}\| \|\Phi_{F^\lambda}(r;\tau)\| dr  =O(\lambda). 
\end{align*}
Then,~\eqref{eq:Strategiesnocosts},~\eqref{eq:frictional position} and the definition of $\Phi$ from~\eqref{eq:matrix exponential} give
\begin{align*}
& d\left[\Phi^\top(t)\left(I_{N-1}\otimes \bar\Lambda^{1/2}\right)\left(\varphi_ t - \bar\varphi_t\right) \right]
\\& = \Phi^\top(t)\left(I_{N-1}\otimes \bar\Lambda^{1/2}\right)\left(\lambda^{-1}\left(I_{N-1}\otimes \bar\Lambda^{-1}\right) F(T-t) \left(\varphi_t - \bar{\varphi}_t\right) dt  + d\left(\varphi_ t - \bar\varphi_t\right) \right)
\\&=  -\lambda^{-1}\Phi^\top(t)\left(I_{N-1}\otimes \bar\Lambda^{-1/2}\right)E^\lambda(T-t)W_t dt +\Phi^\top(t) \left(I_{N-1} \otimes \bar\Lambda^{1/2}\left(\alpha\alpha^\top\right)^{-1} \alpha\right) \xi dW_t . 
\end{align*}
Now, we show that the deviation $\varphi - \bar\varphi$ is approximated by the following ($K(N-1)$-dimensional Ornstein-Uhlenbeck process:
\begin{align}\label{aprroximation delta varphi}
d\Delta_t &:= -\lambda^{-1/2} \left( \Gamma^{1/2} \otimes \bar\Lambda^{-1}\left(\bar\Lambda\#\alpha\alpha^\top \right)\right) \Delta_t dt + \left(I_{N-1} \otimes \bar\Lambda^{1/2}\left(\alpha\alpha^\top\right)^{-1} \alpha\right) \xi dW_t
\notag\\
 &=-\lambda^{-1/2}\left(I_{N-1}\otimes \bar\Lambda^{-1}\right) \widehat F \Delta_t dt + \left(I_{N-1} \otimes \bar\Lambda^{1/2}\left(\alpha\alpha^\top\right)^{-1} \alpha\right) \xi dW_t, \qquad \Delta_0 = 0.
\end{align}
First, again by our matrix-version of variation of constants, we can have the explicit solution of the SDE~\eqref{aprroximation delta varphi} can be written as
$$
\Delta_t = \int_0^t \left(I_{N-1}\otimes \bar\Lambda^{-1/2}\right)\Psi_{\widehat{F}^\top}^\top (r;t) \left(I_{N-1}\otimes  \bar\Lambda^{1/2}\left(\alpha\alpha^\top\right)^{-1} \alpha\right) \xi dW_r,
$$
where $\Psi_{\widehat{F}^\top}(r;t) = \Phi_{\widehat{F}^\top}(r)\Phi_{\widehat{F}^\top}^{-1}(t)$, and $\Phi_{\widehat{F}^\top}$ is the solution to the following matrix linear ODE:
$$
\Phi_{\widehat{F}^\top}'(\tau) = \frac{1}{\lambda^{1/2}}\left(I_{N-1}\otimes\bar\Lambda^{-1/2}\right) \widehat{F}^\top \left(I_{N-1}\otimes\bar\Lambda^{-1/2}\right)\Phi_{\widehat F} (\tau), \qquad \Phi_{\widehat F}(0)= I_{K(N-1)}. 
$$
The process $\Delta$ is a Gaussian with mean 0; moreover, all eigenvalues of its covariance matrix are of order $O(\lambda^{1/2})$. As a consequence, $\mathbb{E}\left[\|\Delta_t\|\right]  = O(\lambda^{1/4})$. 

To assess the accuracy of the asserted asymptotic approximation, consider the (rescaled) difference between $\varphi - \bar\varphi$ and $\Delta$:
\begin{align*}
 &d\left[\Phi^\top(t)\left(I_{N-1}\otimes \bar\Lambda^{1/2}\right)\left(\varphi_ t - \bar\varphi_t - \Delta_t\right)\right]
\\&=  \lambda^{-1}{\Phi^\top(t)}\left(I_{N-1}\otimes \bar\Lambda^{-1/2}\right)\left(\left(\lambda^{1/2}\widehat{F} - F^\lambda(T-t)\right)\Delta_t - E^\lambda(T-r)W_t \right)dt. 
\end{align*}
As the initial value of the difference vanishes by assumption, it follows that 
\begin{align*}
 \varphi_ t - \bar\varphi_t - \Delta_t 
= \lambda^{-1}\int_0^t  \Psi^\top(r;t) \left(I_{N-1}\otimes \bar\Lambda^{-1}\right)
\left(\left(\lambda^{1/2}\widehat{F} - F(T-r)\right)\Delta_r - E^{\lambda}(T-r)W_r \right) dr. 
\end{align*}
With similar argument on $\|\Psi\|_{\mathrm{op}}$ and $|| \lambda^{1/2}\widehat{F} - F^\lambda(T-r) ||_{\mathrm{op}}$ as in the approximation of $H^\lambda$, we obtain
\begin{align*}
\|\varphi - \bar\varphi - \Delta\|_{\mathbb{H}^p}
&\leq\left( \int_0^T \mathbb{E}\left[ \| \varphi_ t - \bar\varphi_t - \Delta_t \|_{\mathrm{op}}^{2p} \right] dt \right)^{1/2p} 
\\& \leq \lambda^{-1}\int_0^T \left( \int_0^t \| \Psi(r;t)\|_{\mathrm{op}}
\left\|\lambda^{1/2}\widehat{F} - F^\lambda(T-r)\right\|_{\mathrm{op}} O(\lambda^{1/4} )\; dr + O(\lambda^{3/2})  \right) dt 
\\&= O(\lambda^{1/4} )  \int_0^T \left\|   \Psi (0;t ) \right\|_{\mathrm{op}} dt +O(\lambda^{1/2}) = O(\lambda^{1/2} ).
\end{align*}
Now, recall $\dot\varphi$ from~\eqref{eq:frictional strategy}, which we can rewrite as
\begin{align*}
\dot\varphi_t 
&= - \lambda^{-1} \left(I_{N-1}\otimes \bar\Lambda^{-1}\right)\left[F^\lambda(T-t) \left(\varphi_t - \bar{\varphi}_t\right) + E^\lambda (T-t)W_t\right]
\\&=  -\lambda^{-1/2} \left(I_{N-1}\otimes \bar\Lambda^{-1}\right) \widehat{F} \Delta_t + O_{\mathbb{H}^p}(1) 
\\&=  -\lambda^{-1/2} \left( \Gamma^{1/2} \otimes \bar\Lambda^{-1}\left(\bar\Lambda\#\alpha\alpha^\top\right) \right)\Delta_t + O_{\mathbb{H}^p}(1).
\end{align*}
Setting $\dot{\bar{\varphi}} :=-\lambda^{-1/2} \left( \Gamma^{1/2} \otimes \bar\Lambda^{-1}\left(\bar\Lambda\#\alpha\alpha^\top\right) \right)\Delta$, we then have
\begin{align}\label{approx:trading volume}
d\dot{\bar{\varphi}}_t &= -\lambda^{-1/2} \left( \Gamma^{1/2} \otimes \bar\Lambda^{-1}\left(\bar\Lambda\#\alpha\alpha^\top\right) \right) d\Delta_t \notag\\
& = -\lambda^{-1/2} \left( \Gamma^{1/2} \otimes \bar\Lambda^{-1}\left(\bar\Lambda\#\alpha\alpha^\top\right) \right)\left(\dot{\bar{\varphi}}_t dt +  \left(I_{N-1} \otimes \bar\Lambda^{1/2}\left(\alpha\alpha^\top\right)^{-1} \alpha\right) \xi dW_t\right),
\end{align}
which established the desired approximation from Theorem~\ref{thm:asymp}. 

To derive the corresponding result for the equilibrium prices, recall from~\eqref{eq:frictional position}-\eqref{eq:Ybis} in Theorem~\ref{thm:Radnerfric} that the difference of frictional and frictionless price level is
\begin{align*}
S_t- \bar{S}_t 
&= \mathcal{Y}_t - {\lambda}\left(c\otimes \bar\Lambda\right)^\top \dot{\varphi}_t = \mathcal{Y}_t +\lambda^{1/2} \left(c^\top \Gamma^{1/2} \otimes \left(\bar\Lambda\#\alpha\alpha^\top\right) \right)\Delta_t + O_{\mathbb{H}^p}(\lambda).
\end{align*}
At the initial time $t=0$, $\Delta_0=0$, the definition of $\mathcal{Y}$ in~\eqref{eq:Ybis} and the estimates~\eqref{approximation H} from Lemma~\ref{bounds} give
\begin{align*}
&S_0- \bar{S}_0 \\
&=
-\bar{\gamma} \left( \int_0^{T}  \left({c} \otimes I_K\right)^\top H^\lambda(r) \alpha^\top + \alpha \left(H^\lambda (r) \right)^\top\left({c} \otimes I_K\right) +\left({c} \otimes I_K\right)^\top H^\lambda(r)\left( H^\lambda(r)\right)^\top  \left({c} \otimes I_K\right) dr  \right) s
\\&= -\bar{\gamma} \left( \int_0^{T}  \left({c} \otimes I_K\right)^\top H^\lambda (r) \alpha^\top + \alpha \left(H^\lambda(r)\right)^\top\left({c} \otimes I_K\right)  dr  \right) s +O_{\mathbb{H}^p}(\lambda)
\\&= -\bar{\gamma} \left( \int_0^{T} \left( \left({c} \otimes I_K\right)^\top H^\lambda (r)- \lambda^{1/2} M\xi \right)\alpha^\top + \alpha \left(H^\lambda(r)- \lambda^{1/2} M\xi\right)^\top\left({c} \otimes I_K\right)  dr \right) s \\ 
&\quad - \lambda^{1/2} \bar{\gamma} \left( M\xi \alpha^\top + \alpha \xi^\top M^\top\right) s T + O_{\mathbb{H}^p}(\lambda)
\\& =- \lambda^{1/2} \bar{\gamma} \left( M\xi \alpha^\top + \alpha \xi^\top M^\top\right) s T + O_{\mathbb{H}^p}(\lambda).
\end{align*}
A straightforward but tedious computation shows that the drift term of $\mathcal{Y}$ from~\eqref{eq:Ybis} can be written as 
\begin{align*}
\frac{d\mathcal{Y}_t}{dt} &=\bar{\gamma} \left(\left({c} \otimes I_K\right)^\top H^\lambda \alpha^\top + \alpha \left(H^\lambda  \right)^\top\left({c} \otimes I_K\right) +\left({c} \otimes I_K\right)^\top H^\lambda \left( H^\lambda\right)^\top  \left({c} \otimes I_K\right)(T-t) \right) s \\
&=\lambda^{1/2} \bar{\gamma} \left( M\xi \alpha^\top + \alpha \xi^\top M^\top\right) s + O_{\mathbb{H}^p}(\lambda).
\end{align*}
For the drift term of $\dot{\varphi}$ from~\eqref{approx:trading volume} we have
$$
 -\lambda^{-1/2} \left( \Gamma^{1/2} \otimes \bar\Lambda^{-1}\left(\bar\Lambda\#\alpha\alpha^\top\right) \right)\dot{\bar{\varphi}}+ O_{\mathbb{H}^p}(1).
$$
In summary, we therefore obtain the following approximation for the frictional equilibrium expected returns:
\begin{align*}
\mu_t 
&= \bar\mu_t +
 \lambda ^{1/2} \left(c^\top\Gamma^{1/2} \otimes \left(\bar\Lambda\#\alpha\alpha^\top \right) \right) \dot{\varphi}_t+ \lambda^{1/2}\bar{\gamma} \left(M \xi\alpha^\top  + \alpha\xi^\top M^\top  \right) s +O_{\mathbb{H}^p}(\lambda)\\
&=\bar\gamma \alpha \alpha^\top s + \lambda ^{1/2} \left(c^\top\Gamma^{1/2} \otimes \left(\bar\Lambda\#\alpha\alpha^\top \right) \right) \dot{\varphi}_t+ \lambda^{1/2}\bar{\gamma} \left(M \xi\alpha^\top  + \alpha\xi^\top M^\top  \right) s +O_{\mathbb{H}^p}(\lambda).
\end{align*}
This completes the proof.
\end{proof}

\bibliographystyle{abbrv}
\interlinepenalty=10000
\bibliography{references}

\end{document}